\pdfoutput=1
\documentclass[12pt, draftclsnofoot, onecolumn]{IEEEtran}
\usepackage{cite}
\usepackage{indentfirst}
\usepackage{graphicx}
\usepackage{changepage}
\usepackage{stfloats}
\usepackage{amsfonts,amssymb}
\usepackage{bm}
\usepackage{algorithm}
\usepackage{algorithmic}
\usepackage{amsmath}
\usepackage{amsmath,amssymb,amsfonts}
\usepackage{times}
\usepackage{url}
\usepackage{cite}
\usepackage{textcomp}
\usepackage{xcolor}
\usepackage{color}
\usepackage{setspace}
\usepackage{enumitem}
\usepackage{float}
\usepackage{diagbox}
\usepackage[T1]{fontenc}
\usepackage[utf8]{inputenc}
\usepackage{authblk}
\usepackage{threeparttable}
\usepackage{booktabs}
\usepackage{footnote}
\usepackage{graphicx}
\usepackage{pifont}
\usepackage{subfigure}
\usepackage{mathrsfs}
\allowdisplaybreaks[4]
\newenvironment{proof}{{\indent  \indent \it Proof:}}{\hfill $\blacksquare$}

\begin{document}
\title{Intelligent Reflecting Surface Enabled Multi-Target Sensing}
\begin{spacing}{1.4}

\author{
	Kaitao Meng, \textit{Member, IEEE}, Qingqing Wu, \textit{Senior Member, IEEE}, Robert Schober, \textit{Fellow, IEEE}, and Wen Chen, \textit{Senior Member, IEEE}
	\thanks{{K. Meng is with the State Key Laboratory of Internet of Things for Smart City, University of Macau, Macau, 999078, China. }{(emails: \{kaitaomeng\}@um.edu.mo)}. Q. Wu is with Shanghai Jiao Tong University, Shanghai 201210, also with the State Key Laboratory of Internet of Things for Smart City, University of Macau, Macau, 999078, China. {(emails: \{qingqingwu\}@um.edu.mo) Robert Schober is with the Institute for Digital Communications, Friedrich-Alexander University Erlangen-Nürnberg (FAU), 91054 Erlangen, Germany (e-mail: robert.schober@fau.de). W. Chen is with the Department of Electronic Engineering, Shanghai Jiao Tong University, Shanghai 201210, China (email: wenchen@sjtu.edu.cn).} } 
}
\maketitle

\vspace{-6mm}
\begin{abstract}
	Besides improving communication performance, intelligent reflecting surfaces (IRSs) are also promising enablers for achieving larger sensing coverage and enhanced sensing quality. Nevertheless, in the absence of a direct path between the base station (BS) and the targets, multi-target sensing is generally very difficult, since IRSs are incapable of proactively transmitting sensing beams or analyzing target information. Moreover, the echoes of different targets reflected via the IRS-assisted virtual links arrive at the BS from the same direction. In this paper, we study a wireless system comprising a multi-antenna BS and an IRS for multi-target sensing, where the beamforming vector and the IRS phase shifts are jointly optimized to improve the sensing performance. To meet the different sensing requirements, such as a minimum received power and a minimum sensing frequency, we propose three novel IRS-assisted sensing schemes: \textit{Time division (TD) sensing}, \textit{signature sequence (SS) sensing}, and \textit{hybrid TD-SS sensing}. For TD sensing, the sensing tasks are performed in sequence over time. In contrast, the novel SS sensing scheme senses all targets simultaneously and establishes a relationship between the target directions and SSs. To strike a flexible balance between the beam pattern gain and sensing efficiency, we also propose a general hybrid TD-SS sensing scheme with target grouping, where targets belonging to the same group are sensed simultaneously via SS sensing, while the targets in different groups are assigned to orthogonal time slots. By controlling the number of groups, hybrid TD-SS sensing can provide a more flexible balance between beam pattern gain and sensing frequency. Moreover, we propose a two-layer penalty-based algorithm to solve the challenging non-convex optimization problem for the joint design of the BS beamformers, IRS phase shifts, and target grouping. Simulation results demonstrate the effectiveness of the proposed hybrid scheme in achieving a flexible trade-off between beam pattern gain and sensing frequency. Our results also reveal that the power leakage in unintended directions is larger for tighter interference constraints.
\end{abstract}   

\newtheorem{thm}{\bf Lemma}
\newtheorem{remark}{\bf Remark}
\newtheorem{Pro}{\bf Proposition}
\newtheorem{theorem}{Theorem}

\section{Introduction}

With the development of emerging environment-aware applications such as autonomous driving, smart manufacturing, and disaster monitoring, both communication and sensing services have to meet more stringent requirements in next-generation wireless networks \cite{liu2021integrated, Mu2022NOMA}. Fortunately, benefiting from advancements in millimeter-wave (mmWave)/terahertz (THz) and multiple-input multiple-output (MIMO) technologies, future base stations (BSs) can not only provide enhanced-quality and low-latency communication but also enable high-resolution and high-accuracy sensing via narrow highly directional pencil-like beams  \cite{Liu2020JointRadar, Cui2021Integrating, LiuX2020Joint, Meng2022UAV}. However, the requirements for sensing are not the same as those for communication due to the different performance metrics and information acquisition methods. Specifically, with a proper pilot training design and channel estimation, both line-of-sight (LoS) and non-LoS (NLoS) paths can be exploited for data transmission; in contrast, for sensing, typically only LoS paths are exploited while NLoS paths are treated as unfavorable interference for target detection and parameter estimation \cite{TargetDetectionLocalization, sume2011radar}. Thus, one critical issue for sensing via multi-antenna BSs is that their coverage may be practically limited due to blocked areas caused by obstacles, especially in urban environments with many potential obstructions \cite{Solomitckii2021Radar}. 
\par
Recently, intelligent reflecting surfaces (IRSs) have been proposed as a promising technique to achieve larger communication coverage and improved transmission quality \cite{liu2020matrix, wu2019beamformingDiscrete, you2020channel, Hua2022Novel, chen2021irs}. By exploiting the adaptability of the amplitudes and/or phase shifts of the uniformly placed passive IRS elements, a controllable wireless propagation environment with increased degrees of freedom (DoFs) can be realized, thereby enhancing signal reception at a designated receiver and avoiding interference to other users\cite{Abeywickrama2020Intelligent, hua2022joint}. Besides improving communication performance, IRSs can also be exploited by multi-antenna radar BSs to increase the signal power reflected from the target; IRSs can aid in establishing artificial virtual LoS links from the BS to targets in the absence of direct paths between them \cite{Fang2021SINR, zheng2021double}. Accordingly, by leveraging their controllability, IRSs can be utilized for detection of nearby targets for security management (e.g., eavesdropper detection and unmanned aerial vehicle (UAV) monitoring) \cite{Lu2021intelligent, Lu2021Target} and localization of legitimate devices for signaling overhead reduction \cite{Zhang2021MetaLocalization}. In \cite{buzzi2021foundations}, an IRS is employed to enhance target detection performance by optimizing its phase shifts to concentrate the radiated power at specific locations. Furthermore, IRS phase shifts and BS beamformers have been jointly optimized for radar sensing systems to facilitate IRS-assisted sensing \cite{buzzi2021radar}. This joint beamforming design can also benefit integrated sensing and communication (ISAC) systems \cite{Prasobh2021Joint}. However, the systems considered in \cite{buzzi2021foundations, buzzi2021radar, Prasobh2021Joint} assume a single-target scenario, and the developed schemes may not be directly applicable to the multi-target case.  
\par
In fact, IRS-assisted multi-target sensing in blocked areas is very challenging due to the IRS's limited beam design and signal processing capabilities. In \cite{zhang2022metaradar}, the radar waveform and the IRS phase shifts were jointly optimized to improve multi-target detection accuracy in radar systems. However, this scheme can only be utilized when BS and IRS are close to each other, limiting the improvement in sensing coverage. In a recent work \cite{Aubry2021Reconfigurable}, the authors studied IRS-aided radar surveillance in NLoS scenarios, where a sequential scanning protocol was adopted to cover the NLoS area. This scanning protocol incurs low hardware cost but may cause a low pulse repetition frequency (PRF) \cite{Rambach2013Colocated} and asynchronous estimation errors for multi-target scenarios \cite{Hügler2018Radar}, \cite{Hakobyan2019High}. An IRS-assisted ISAC system was investigated in \cite{song2021joint}, where the beam pattern gain of the IRS was optimized for target sensing, while accounting for communication requirements. However, careful beam pattern gain design alone may not ensure reliable IRS-assisted target sensing in blocked areas due to the following reasons. First, direction-of-arrival (DOA) estimation has to rely on the reflected echoes received at the BS \cite{aubry2021ris}, since the IRS is incapable of proactively transmitting beams or analyzing target information. Second, since the BS-IRS link is shared by all targets, the echoes of different targets reflected by the IRS arrive from the same direction at the BS, thus making DOA estimation very challenging  \cite{Luzhou2006Radar, Stoica2007Probing}. Therefore, in practice, effective simultaneous multi-target sensing via IRS-assisted virtual LoS links is difficult.
\par 
Motivated by the above discussion, we study an IRS-assisted MIMO radar system, where one IRS and one multi-antenna BS are jointly designed to sense multiple targets in a blocked area. To meet the different sensing requirements, such as a minimum received power and a minimum sensing frequency,\footnote{Sensing frequency refers to the number of times that the same target can be sensed in one second.} we propose three novel IRS-assisted sensing schemes: \textit{Time division (TD) sensing}, \textit{signature sequence (SS) sensing}, and \textit{hybrid TD-SS sensing}. First, for \textit{TD sensing}, by sequentially sensing each target in orthogonal time slots, all IRS elements can be exploited to focus the beam pattern on one target at a time and interference between different targets is avoided. The limited sensing frequency of \textit{TD sensing} can be overcome by the second scheme, \textit{SS sensing}, which enables simultaneous multi-target sensing. Different from conventional radar systems that distinguish the targets based on the direction from which the corresponding echos arrive, we propose to modulate the sensing beams with dedicated SSs to establish a relationship between the directions/targets and the SSs. The power leakage of the radar/sensing beams is reduced by introducing suitable constraints on the beam pattern. This facilitates the simultaneous detection of multiple targets by analyzing the received echo signals at the BS, thereby achieving high sensing efficiency. However, due to the potentially limited number of radio frequency (RF) chains and the finite BS transmit power, the number of targets that can be sensed simultaneously is limited. Hence, a more general scheme, \textit{hybrid TD-SS sensing}, is proposed to divide the sensing targets into several disjoint groups, and perform \textit{SS sensing} within each group to improve the intra-group sensing efficiency and \textit{TD sensing} across different groups to avoid inter-group interference. 
\par
For the three proposed sensing schemes, the joint optimization of the beamformer at the BS, the IRS phase shifts, and the target grouping is studied to maximize the beam pattern gains in the targets' directions, while limiting the interference in other directions. However, solving the resulting optimization problem is highly non-trivial due to the closely coupled integer optimization variables and highly non-convex constraints. To overcome these challenges, we propose a two-layer penalty-based algorithm to solve the challenging non-convex optimization problem. The main contributions of this paper can be summarized as follows:
\begin{itemize}
	\item We propose three sensing schemes: \textit{TD sensing}, \textit{SS sensing}, and \textit{hybrid TD-SS sensing}. For \textit{TD sensing}, multiple targets are sensed sequentially in orthogonal time slots. For \textit{SS sensing}, we develop a novel SS beamforming scheme for IRS-assisted multi-antenna BSs, facilitating simultaneous multi-target sensing based on a one-to-one mapping between targets and sensing beams modulated with dedicated SSs. For \textit{hybrid TD-SS sensing}, a more flexible trade-off between beam pattern gain and sensing frequency is realized based on the joint design of the BS beamformers, IRS phase shifts, and target grouping.
	\item The beamformer at the BS and the IRS phase shifts are jointly optimized. For \textit{TD sensing}, the formulated optimization problem is converted to a semidefinite programming (SDP) problem and the rank-one constraint on the IRS phase shift vector is transformed to facilitate phase shift design. For \textit{SS sensing}, we prove that rank-one beamforming at the BS is optimal, and based on this, an inner approximation algorithm is proposed to obtain a locally optimal solution to the joint BS beamformer and IRS phase shift optimization problem with guaranteed convergence to a Karush-Kuhn-Tucker (KKT) point. Also, for the two-target scenario and any given phase shift configuration, we provide the optimal transmit beamforming vector in closed form.
	\item We prove that there is a fundamental trade-off between beam pattern gain and sensing frequency for the proposed \textit{hybrid TD-SS sensing} scheme. The big-M formulation is adopted to decompose the coupled integer optimization variables. Then, a two-layer penalty-based algorithm is proposed to jointly optimize the beamforming vectors at the BS, the IRS phase shifts, and the target grouping. 
	\item Finally, our simulation results verify the trade-off between beam pattern gain and sensing frequency for IRS-assisted sensing systems and validate the superiority of the proposed schemes over benchmark schemes. Our results also reveal that as the interference constraints become stricter and/or the number of targets in each group increases, the maximum beam pattern gain decreases while the power leakage in unintended directions increases. 
\end{itemize}

The remainder of this paper is organized as follows. Section \ref{SYSTEM} introduces the system model and problem formulation for IRS-assisted sensing. In Section \ref{TDandCosensing}, the {\textit{TD}} and \textit{SS sensing} schemes are proposed and analyzed. Section \ref{AnalysisGeneral} presents a penalty-based algorithm for \textit{hybrid TD-SS sensing}. Section \ref{Simulations} provides numerical results to validate the performance of the proposed sensing schemes. Section \ref{Conclusion} concludes this paper.

\textit{Notations}: $\|{\bm{x}} \|$ denotes the Euclidean norm of a complex-valued vector ${\bm{x}}$. ${\rm{diag}}({\bm{x}})$ denotes a diagonal matrix whose main diagonal elements are the elements of ${\bm{x}}$. $[{\bm{x}}]_{n}$ denotes the $n$th element of ${\bm{x}}$. For a general matrix ${\bm{X}}$, $\operatorname{rank}({\bm{X}})$, ${\bm{X}}^H$, ${\bm{X}}^T$, and $[{\bm{X}}]_{m,n}$ respectively denote its rank, conjugate transpose, transpose, and the element in the $m$th row and $n$th column. For a square matrix ${\bm{Y}}$, ${\rm{Tr}}({\bm{Y}})$ and ${\bm{Y}}^{-1}$ denote its trace and inverse, respectively, while ${\bm{Y}} \succeq 0$ indicates that ${\bm{Y}}$ is a positive semidefinite matrix. $\sigma_n({\bm{X}})$ denotes the $n$th eigenvalue of matrix ${\bm{X}}$. $\mathbb{E}(\cdot)$ denotes statistical expectation, and $\{\cdot\}$ represents a variable set. ${\rm{Re}} \left\{{\bm{X}}\right\}$ denotes the real part of matrix ${\bm{X}}$, and ${\cal{O}}(\cdot)$ is the big-O notation for computational complexity. 

\begin{figure}[t]
	\centering
	\setlength{\abovecaptionskip}{0.cm}
	\includegraphics[width=15cm]{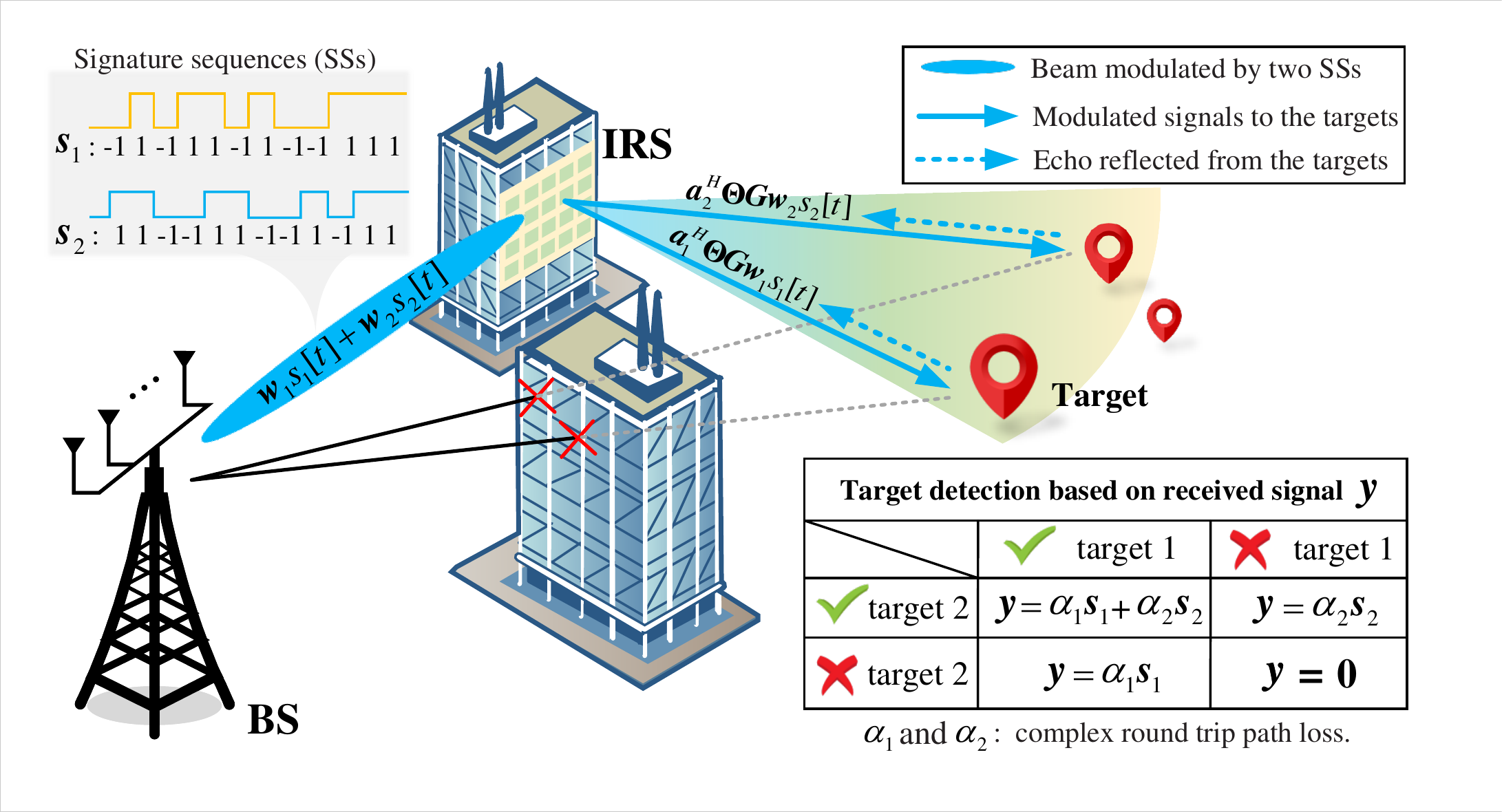}
	\caption{Illustration of IRS-assisted sensing of multiple targets based on a dedicated radar/sensing beam design.}
	\label{figure1}
\end{figure}
\section{Problem Formulation and System Model}
\label{SYSTEM}

\subsection{IRS Enabled Sensing Model}
\label{SpatialCoding}
\par
As shown in Fig. 1, we consider an IRS-assisted multi-antenna BS aiming to detect multiple targets blocked by obstacles, e.g., buildings and billboards. Without loss of generality, we assume that the BS employs a uniform linear array (ULA) having $M$ antenna elements. A uniform planar array (UPA) is equipped at the IRS. The center of the IRS is assumed to be located at the origin of the 3-dimensional (3D) Cartesian coordinate system, while adjacent IRS elements are separated by $d_x$ and $d_y$ in $x$- and $y$-direction, respectively, as shown in Fig.~\ref{figure2}. The number of reflecting IRS elements is denoted by $N = N_x \times N_y$, where $N_x$ and $N_y$ denote the number of elements along the $x$- and $y$-axis, respectively. Moreover, to reduce the design complexity and hardware cost, the reflection amplitude of the IRS elements is assumed to be equal to one \cite{wu2019beamformingDiscrete}. The reflection-coefficient matrix of the IRS is denoted by ${\bm{\Theta}} = {\rm{diag}}(e^{j \theta_{1}}, ... , e^{j \theta_{N}})$, where $j$ denotes the imaginary unit, and $\theta_{n} \in [0, 2\pi)$ is the phase shift of the $n$th IRS element, $n \in {\cal{N}} = \{ 1,\cdots,N\}$. 
\par 
\begin{figure}[t]
	\centering
	\setlength{\abovecaptionskip}{0.cm}
	\includegraphics[width=13.5	cm]{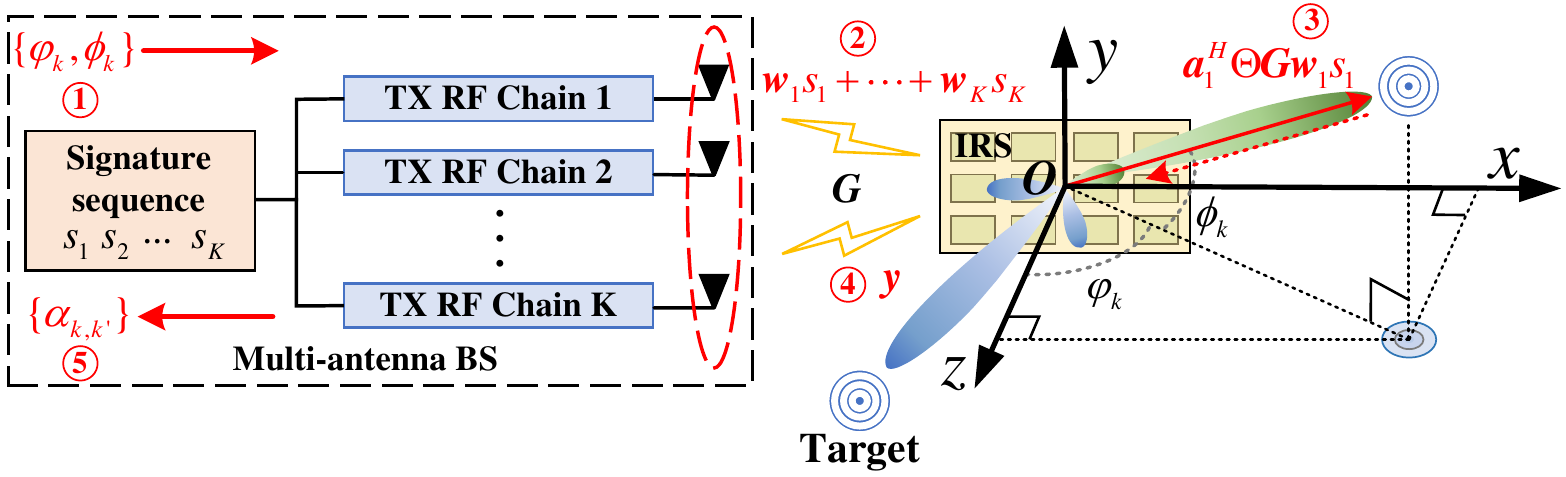}
	\caption{Illustration of SS sensing model for the considered IRS-assisted multi-target sensing scenario. The red arrows and the circles with numbers represent the inputs/outputs and the order of the signal flow, respectively.}
	\label{figure2}
\end{figure}
The directions of the prospective targets are indexed by $k \in {\cal{K}} = \{1,\cdots,K\}$. Specifically, target $k$'s direction is defined by $\{\varphi_k, \phi_k\}$, where $\varphi_k$ and $\phi_k$ are respectively the azimuth and elevation angles of the geometric path connecting the IRS and target $k$, and are assumed to be known from the specific requirements of the underlying sensing task. For target detection, directions $\{\varphi_k, \phi_k\}$ can be chosen as general directions towards a region, where targets are expected, or as the estimated directions of targets based on previous sensing results for target tracking. The transmission link between BS and IRS is characterized by matrix ${\bm{G}} \in {\mathbb{C}}^{M \times N}$, which can be estimated by the methods proposed in \cite{Hu2021Semi, Zheng2020Intelligent}. This channel is slowly varying in practice since the BS and the IRS are deployed at fixed locations. The reflect array response of the IRS in the direction of the $k$th target can be written as \cite{Wei2021Sum}
\begin{equation}
	\begin{aligned}
		{\bm{a}}_{k}
		= \left[1, \cdots, e^{ \frac{-j 2 \pi\left(N_{x}-1\right) {d}_{x} {\Phi}_k }{\lambda}}\right]^T  \otimes\left[1, \cdots, e^{ \frac{ -j 2 \pi\left(N_{y}-1\right) {d}_{y} {\Omega}_k }{\lambda}}\right]^T,
	\end{aligned}
\end{equation}
where ${\Phi}_k \triangleq \sin \left(\varphi_k\right) \cos \left(\phi_k\right)$, ${\Omega}_k \triangleq \sin \left(\varphi_k\right) \sin \left(\phi_k\right)$, $\lambda$ is the carrier wavelength, and $\otimes$ denotes the Kronecker product. 

\subsection{SS Sensing Model}
In the proposed SS sensing scheme, the MIMO BS transmits several dedicated radar beams modulated by SSs ${\bm{s}}_k = [s_k[1],\cdots,s_k[N_p]]^T$, $\forall k$, where $N_p$ denotes the number of pulses and $s_k[t]$ is the SS for sensing target $k$ in the $t$th pulse repetition interval (PRI) \cite{Xiao2022Waveform, Blunt2016Overview}. Exploiting the signatures, a one-to-one mapping between SS ${\bm{s}}_k$ and intended direction $\{\varphi_k, \phi_k\}$ can be established, as illustrated in Fig.~\ref{figure1}. Accordingly, whether there is a target in direction $k$ or not can be determined by detecting the echo corresponding to SS ${\bm{s}}_k$ at the BS; similarly, for target tracking, the parameters of target $k$ within the range-Doppler region of interest can be analyzed at the BS based on the propagation delay and frequency shift of the echoes modulated by SS ${\bm{s}}_k$. To be able to effectively distinguish the echoes corresponding to different dedicated SSs, an orthogonal codebook is utilized, i.e., ${\bm{s}^H_k}{\bm{s}_{k'}} = 0$ and ${\bm{s}^H_k}{\bm{s}_{k}} = N_p$, where $k \ne k'$, $k, k' \in {\cal{K}}$, and $|s_k[t]|^2 = 1$. Linear transmit precoding is applied at the BS, where a dedicated radar/sensing beam is allocated to each target. Hence, in time interval $t$, the complex baseband transmitted signal at the BS can be expressed as a weighted sum of radar beams, i.e.,  
\begin{equation}
	{\bm{x}}[t] = \sum\nolimits_{k = 1}^K{{{\bm{w}}_{k}} {s_{k}}[t]},
\end{equation}
where ${{\bm{w}}_{k}} \in \mathbb{C}^{M \times 1}$ is the transmit beamforming vector for target $k$. Accordingly, the beam pattern gain at the IRS for SS $k$ in direction $\{\varphi_k, \phi_k\}$ can be expressed as follows:
\begin{equation}\label{UserReceivedSignal}
	\mathcal{P}_k = \mathbb{E}\left(\left|\bm{a}^{{H}}_k \bm{\Theta} \bm{G} {{{\bm{w}}_{k}} {s_{k}}[t]}\right|^{2}\right) = \bm{a}^{{H}}_k \underbrace {\bm{\Theta}  \bm{G}{{\bm{w}}_{k}}{{\bm{w}}^H_{k}} \bm{G}^{H} \bm{\Theta}^{{H}}}_{\text{covariance matrix at the IRS}} \bm{a}_k.
\end{equation}
Furthermore, the received signal at the BS in time interval $t$ is given by
\begin{equation}\label{BeampatternGain}
	\begin{aligned}
		&{\bm{y}}[t] =\sum\nolimits_{k = 1}^K   \underbrace{\bm{G}^{H} \bm{\Theta}^{{H}} \bm{a}_{k} \beta_{k} \bm{a}^{{H}}_{k} {\bm{\Theta}  \bm{G}{{\bm{w}}_{k}}  {s_{k}}[t - \tau_{k}]}}_{\text{SS ${\bm{s}}_k$ reflected from target $k$}}  \\
		+&  \sum\nolimits_{k = 1}^K \underbrace{\sum\nolimits_{k' \in {\cal{K}} \backslash k } \bm{G}^{H} \bm{\Theta}^{{H}} \bm{a}_{k'} \beta_{k'} \bm{a}^{{H}}_{k'} {\bm{\Theta}  \bm{G}{{\bm{w}}_{k}}  {s_{k}}[t - \tau_{k'}]}}_{\text{SS ${\bm{s}}_k$ reflected from unassociated targets}} + {\bm{n}}[t],
	\end{aligned}
\end{equation}
where $\beta_{k}$ represents the reflection coefficient of the $k$th target multiplied by the round-trip path loss of the IRS-target-IRS link, $\tau_k$ is the round-trip propagation delay of signals reflected from target $k$, which is assumed to have an integer value for simplicity of presentation, ${\bm{n}}[t]$ is additive white Gaussian noise (AWGN) with covariance matrix $\sigma^2 {\bm{I}}_{M}$, and ${\bm{I}}_{M}$ is an identity matrix of size $M$. Unlike traditional beam pattern designs \cite{Liu2020JointRadar, LiuX2020Joint}, for the proposed SS sensing, the beam pattern gain of the radar beam intended for target $k$ towards other targets is considered to cause undesired interference, see (\ref{BeampatternGain}). The signal energy received from undesired directions may result in poor detection performance. To limit the interference induced by the power leakage of the dedicated radar/sensing beams in other targets' directions, e.g., due to overlapping main lobes or high side lobes of different beams, interference constraints are introduced as follows:
\begin{equation}\label{CrossCorrelationfirst}
	\sum\nolimits_{k'\in {\cal{K}} \backslash k } \mathbb{E}\left(\left|\bm{a}^{{H}}_{k'} \bm{\Theta} \bm{G} {{{\bm{w}}_{k}} {s_{k}}}\right|^{2}\right) = \sum\nolimits_{k'\in {\cal{K}} \backslash k }  \bm{a}^{{H}}_{k'} \bm{\Theta}  \bm{G}{{\bm{w}}_{k}}{{\bm{w}}^H_{k}} \bm{G}^{H} \bm{\Theta}^{{H}} \bm{a}_{k'} \le \varepsilon, \forall k,
\end{equation}
where $\varepsilon$ denotes the interference threshold.\footnote{$\varepsilon$ should be much lower than the beam pattern gain in the desired directions, and may be set one to two orders of magnitude lower than the maximum beam pattern gain.} For $\varepsilon = 0$, conventional zero forcing (ZF)-based beamforming is realized. We note that different from the cross-correlation pattern constraints in \cite{Stoica2007Probing}, for SS sensing, we adopt several dedicated radar beams to simultaneously sense multiple directions via the IRS, as illustrated in Fig.~\ref{figure1}, and maximize the beam pattern gain towards the intended directions while limiting the interference caused by the beams in other directions. As a consequence, the received signals at the BS in time interval $t$ can be modeled as
\begin{equation}\label{SimplifiedBeampatternGain}
{\bm{y}}[t] =\sum\nolimits_{k = 1}^K   \bm{G}^{H} \bm{\Theta}^{{H}} \bm{a}_{k} \beta_{k} \bm{a}^{{H}}_{k} {\bm{\Theta}  \bm{G}{{\bm{w}}_{k}}  {s_{k}}[t - \tau_{k}]}  + \tilde{\bm{n}}[t],
\end{equation}
where $\tilde{\bm{n}}[t] \in \mathbb{C}^{M \times 1}$ represents the residual interference plus noise. In this case, the optimal receiver filter ${\bm{f}}_{k}^H$ maximizing the signal power reflected from target $k$ is given by ${\bm{f}}_{k}^H = \frac{(\bm{G}^{H} \bm{\Theta}^{{H}} \bm{a}_{k} \beta_{k} \bm{a}^{{H}}_{k} {\bm{\Theta}  \bm{G}{{\bm{w}}_{k}}})^H}{\| (\bm{G}^{H} \bm{\Theta}^{{H}} \bm{a}_{k} \beta_{k} \bm{a}^{{H}}_{k} {\bm{\Theta}  \bm{G}{{\bm{w}}_{k}})} \|}$. As a result, the combined signal for target $k$ in time interval $t$ is given by 
\begin{equation}
	{\tilde{y}}_k[t] = {\bm{f}}_{k}^H \left(  \sum\nolimits_{{i} = 1}^K \bm{G}^{H} \bm{\Theta}^{{H}} \bm{a}_{i} \beta_{i} \bm{a}^{{H}}_{i} \bm{\Theta}  \bm{G}{{\bm{w}}_{i}}  {s_{i}}[t - \tau_{i}] +  \tilde{\bm{n}}[t]\right).
\end{equation}
Let $\alpha_{k,i} = {\bm{f}}_{k}^H \bm{G}^{H} \bm{\Theta}^{{H}} \bm{a}_{i} \beta_{i} \bm{a}^{{H}}_{i} \bm{\Theta}  \bm{G}{{\bm{w}}_{i}}$ and $\tilde{\bm{y}}_k = [{\tilde{y}}_k[1 + \tau_k],\cdots,{\tilde{y}}_k[N_p + \tau_k]]^T$. By multiplying $\tilde{\bm{y}}_k$ with ${\bm{s}}^H_k$, the component corresponding to target $k$ can be extracted from the received signal $\tilde{\bm{y}}_k$, i.e., 
\begin{equation}\label{TransformationReceiver}
		z_k = {\bm{s}}^H_{k} \tilde{\bm{y}}_k = \sum\nolimits_{{i} = 1}^K \alpha_{k,{i}} {\bm{s}}^H_{k} {\bm{s}}_{i} +  {\bm{f}}_{k}^H \sum\nolimits_{t = 1}^{N_p}  {s_{k}}[t]  \tilde{\bm{n}}[t + \tau_k] \overset{(a)}{=} \alpha_{k,{k}} {\bm{s}}^H_{k} {\bm{s}}_{k} +  \bar {n}_k,
\end{equation}
where $\bar {n}_k = {\bm{f}}_{k}^H  \sum\nolimits_{t = 1}^{N_p}  {s_{k}}[t] \tilde{\bm{n}}[t+\tau_k]$. In (\ref{TransformationReceiver}), ($a$) holds since ${\bm{s}^H_k}{\bm{s}_{i}} = 0$ for $k \ne i$. Similar to the key idea of code-division multiple access (CDMA) to distinguish different communication users, in the proposed sensing scheme, SSs are used to separate targets.\footnote{Since the propagation time of the IRS-target-IRS link is practically short, propagation time differences to different targets are assumed to be within one symbol interval. Furthermore, for larger propagation time differences, a large area synchronized (LAS)-CDMA code can be utilized to avoid interference within a maximum allowed time difference \cite{wei2006uplink}. However, in this case, the maximum number of targets that can be sensed will be reduced since the number of orthogonal SSs decreases.} By establishing a one-to-one mapping between target directions $\{\{\phi_k,\varphi_k\}\}_{k=1}^{K}$ and SSs $\{{\bm{s}}_k\}_{k=1}^{K}$, the echoes reflected by different targets become distinguishable. Furthermore, given the desired sensing directions, the BS beamformer and the IRS phase shifts can be jointly designed to focus the power of SS ${\bm{s}}_k$ towards the corresponding target, as shown in Fig.~2. The extracted SS ${\bm{s}}_k$ reflected from target $k$ can be exploited for target detection. For example, a binary hypothesis test may be formulated based on (\ref{TransformationReceiver}). Specifically, the target detection problem in the proposed scheme can be formulated as $K$ composite binary hypothesis tests \cite{richards2014fundamentals}. The binary hypothesis test for target $k$ is given as follows:

${\mathcal{H}}_{k}^0$ : No target in the $k$th direction;

${\mathcal{H}}_k^1$ : A target exists in the $k$th direction.

The detection model can be described by the following hypothesis testing problem \cite{Fishler2006Spatial}.
\vspace{-2mm}  
\begin{equation}
	z_k=\left\{\begin{array}{l}
		\mathcal{H}^{0}_k: \bar {n}_k  \\
		\mathcal{H}^{1}_k: \alpha_{k,k} {\bm{s}}^H_k {\bm{s}}_k + \bar {n}_k
	\end{array}\right.,
\vspace{-2mm}  
\end{equation}
Then, assuming $\bar {n}_k$ is Gaussian distributed, the optimal detector is given by $E=|z_k|^{2} \underset{\mathcal{H}^{0}_k}{\overset{\mathcal{H}^{1}_k}{\lessgtr}} \mu$ \cite{Fishler2006Spatial}, where the decision threshold $\mu$ is set to meet the desired false alarm rate. The reliability of the decision increases with increasing $|\alpha_{k,{k}}|$ and decreasing variance of $\bar {n}_k$. Therefore, the beam pattern gains of the dedicated beams in the respective target directions should be maximized while limiting the power leakage in the directions of the other targets.

\subsection{Problem Formulation}
In this paper, we aim to maximize the minimum beam pattern gains in the target directions by optimizing the IRS phase shifts and the beamforming vectors at the BS, subject to sensing interference and maximum transmit power constraints. 

\subsubsection{{TD sensing}} By performing each sensing task in a separate time slot, interference between different radar beams can be avoided and the entire beam pattern gain can be focused in the desired direction. Then, the beamforming vector and IRS phase shift optimization problem can be decoupled into $K$ sub-problems, i.e.,
\vspace{-2mm}  
\begin{subequations}\label{P1-TD}
\begin{align}
	(\rm{P1\!-\!TD}): \quad & \begin{array}{*{20}{c}}
		\label{P1-TMD-a}\mathop {\max }\limits_{{\bm{w}}_{k}, {\bm{v}}_k} \quad \mathcal{P}_k 
	\end{array}  \\ 
	\mbox{s.t.}\quad
	\label{P1-TMD-b}& \| {\bm{w}}_{k} \|^2 \le P^{\max}, \\
	\label{P1-TMD-c}& |[{\bm{v}}_k]_{n}| = 1, \forall n \in {\cal{N}},  
	\vspace{-2mm}  
\end{align} 
\end{subequations}
where $[{\bm{v}}_k]_{n} = e^{-j\theta_{k,n}}$, ${\bm{v}}_k^H = [e^{-j\theta_{k,1}},\cdots,e^{-j\theta_{k,N}}] \in \mathbb{C}^{1 \times N}$, $\mathcal{P}_k = \bm{a}^{{H}}_k \bm{\Theta}  \bm{G}{{\bm{w}}_{k}}{{\bm{w}}^H_{k}} \bm{G}^{H} \bm{\Theta}^{{H}} \bm{a}_k = {\bm{v}}^H_k {\bm{Q}}_k {{\bm{w}}_{k}} {{\bm{w}}^H_{k}} {\bm{Q}}^H_k {\bm{v}}_k $, and ${\bm{Q}}_k = {\rm{diag}}({\bm{a}}^H_k){\bm{G}} \in \mathbb{C}^{N \times M}$. For problem (\rm{P1-TD}), constraint ({\ref{P1-TMD-b}}) limits the maximum transmit power, and (\ref{P1-TMD-c}) is the unit-modulus constraint for the IRS phase shifts. 

\subsubsection{SS Sensing}
By associating SS ${\bm{s}}_k$ with desired target $k$, all sensing tasks can be performed simultaneously. Then, our objective is to maximize the minimum beam pattern gain towards the targets, which can be formulated as follows:
\vspace{-2mm}  
\begin{subequations}\label{P1-Co}
\begin{align}
	(\rm{P1\!-\!SS}): \quad \label{P1-SS-a}& \begin{array}{*{20}{c}}
		\mathop {\max }\limits_{\{{\bm{w}}_{k}\}, {\bm{v}}} \quad \mathop {\min }\limits_{k \in {\cal{K}}} \quad \mathop  \mathcal{P}_k
	\end{array}  \\ 
	\mbox{s.t.}\quad
	\label{P1-SS-b}& \sum\nolimits_{k'\in {\cal{K}} \backslash k } {\bm{v}}^H {\bm{Q}}_{k'} {\bm{w}}_{k}  {\bm{w}}^H_{k} {\bm{Q}}^H_{k'} {\bm{v}} \le \varepsilon,   \forall k \in {\cal{K}}, \\
	\label{P1-SS-c}& \sum\nolimits_{k = 1}^K \| {\bm{w}}_{k} \|^2 \le P^{\max},   \\
	\label{P1-SS-d}& |[{\bm{v}}_k]_{n}| = 1, \forall n \in {\cal{N}}, 
	\vspace{-2mm}  
\end{align} 
\end{subequations}
where ${\bm{v}}^H = [e^{-j\theta_{1}},\cdots,e^{-j\theta_{N}}]$. In problem (\rm{P1-SS}), constraint (\ref{P1-SS-b}) ensures that the sum interference of the radar signal indented for target $k$ does not exceed the given threshold $\varepsilon$ in the directions of the other targets. 

It is challenging to optimally solve non-convex problems (\rm{P1-TD}) and (\rm{P1-SS}) since the optimization variables are coupled in the objective function and the constraints. To tackle this issue, in the next section, the properties of both problems are investigated first, and then, efficient algorithms are proposed to obtain high-quality solutions. 

\section{TD and SS Sensing}
\label{TDandCosensing}
\subsection{TD Sensing}
\label{TDSensingMechanism}
In this case, the targets are sensed separately via TD, and maximum ratio transmission (MRT) is the optimal beamforming scheme \cite{Lu2021Aerial}, i.e., ${\bm{w}}^*_{k} = \frac{\sqrt{P^{\max}}{\bm{v}}^H_k {\bm{Q}}_{k}}{{\|{\bm{v}}^H_k {\bm{Q}}_{k}\|}}$. Thus, the optimal reflection-coefficient vector ${\bm{v}}_k$ can be obtained by maximizing the equivalent channel gain $\|{\bm{v}}^H_k {\bm{Q}}_{k}\|$, i.e., problem (P1-TD) simplifies to (after dropping index $k$)
\vspace{-2mm}  
\begin{equation}\label{P1-TDSub}
	({\rm{P1-TD.1}}): \quad  \begin{array}{*{20}{c}}
		\mathop {\max }\limits_{{\bm{v}}} \quad \|{\bm{v}}^H {\bm{Q}}\|
	\end{array} \quad \mbox{s.t.}\quad  |[{\bm{v}}]_{n}| = 1, \forall n \in {\cal{N}}. 
\vspace{-2mm}  
\end{equation}
Since ${\bm{V}} = {\bm{v}}{\bm{v}}^H$, satisfying ${\bm{V}} \succeq 0$ and ${\rm{rank}}({\bm{V}}) = 1$. Problem (\rm{P1-TD.1}) is equivalent to the following problem:
\vspace{-2mm}  
\begin{subequations}\label{P1-TDSub2}
\begin{align}
	(\rm{P1-TD.2}): \quad \label{P1-TDSub2-a}& \begin{array}{*{20}{c}}
		\mathop {\max }\limits_{{\bm{V}}} \quad {\rm{Tr}}\left({\bm{Q}}{\bm{Q}}^H {\bm{V}} \right)
	\end{array} \\ 
	\mbox{s.t.}\quad
	\label{P1-TDSub2-b}&  |[{\bm{V}}]_{n,n}| = 1, \forall n \in {\cal{N}}, \\
	\label{P1-TDSub2-c}&  {\rm{rank}}({\bm{V}}) = 1, {\bm{V}} \succeq 0.
	\vspace{-2mm}  
\end{align} 
\end{subequations}
Since $\operatorname{Tr}({\bm{V}}) = \sum\nolimits_{n=1}^N \sigma_n({\bm{V}}) \ge \mathop {\max}\limits_n \sigma_n({\bm{V}}) = \|{\bm{V}}\|_{2}$, by adopting the difference-of-convex-functions method \cite{Yang2020Federated}, the non-convex rank-one constraint in (\ref{P1-TDSub2-c}) can be equivalently expressed as 
\vspace{-2mm}  
\begin{equation}\label{RelaxingRankoneConstraintsA}
	\operatorname{Tr}({\bm{V}})-\|{\bm{V}}\|_{2} = 0,
	\vspace{-2mm}  
\end{equation}
where $\|{\bm{V}}\|_2$ represents the spectral norm of matrix ${\bm{V}}$. Then, the 
non-convex constraint (\ref{RelaxingRankoneConstraintsA}) can be added into the objective function of problem (P1-TD.2) as a penalty term \cite{bertsekas1997nonlinear}, i.e.,
\vspace{-2mm}  
\begin{subequations}\label{P1-TDSub3}
	\begin{align}
		(\rm{P1-TD.3}): \quad \label{P1-TDSub3-a}& \begin{array}{*{20}{c}}
			\mathop {\max }\limits_{{\bm{V}}} \quad {\rm{Tr}}\left({\bm{Q}}{\bm{Q}}^H {\bm{V}} \right) - \frac{1}{\eta}\left(\operatorname{Tr}({\bm{V}})-\|{\bm{V}}\|_{2}\right)
		\end{array} \\ 
		\mbox{s.t.}\quad
		\label{P1-TDSub3-b}&  |[{\bm{V}}]_{n,n}| = 1, \forall n \in {\cal{N}},  {\bm{V}} \succeq 0,
		\vspace{-2mm}  
	\end{align} 
\end{subequations}
where $\eta$ denotes the penalty coefficient. When $\frac{1}{\eta} \to \infty$, the solution of (P1-TD.3) satisfies the equality constraint (\ref{RelaxingRankoneConstraintsA}), and hence, is rank-one and also solves (P1-TD.2). Since $\|{\bm{V}}\|_{2}$ is convex with respect to (w.r.t.) ${\bm{V}}$, the first-order Taylor approximation of $\|{\bm{V}}\|_{2}$ is given as follows:
\vspace{-2mm}  
\begin{equation}\label{RelaxingRankoneConstraintsB}
	\|{\bm{V}}\|_{2} \ge \operatorname{Tr}\left(\boldsymbol{u}_{\max }\left({\bm{V}}^{(r)}\right)\right. \left. \boldsymbol{u}_{\max }^{H}\left({\bm{V}}^{(r)}\right)\left({\bm{V}}-{\bm{V}}^{(r)}\right)\right)-\left\|{\bm{V}}^{(r)}\right\|_{2} = {\cal{U}}^{(r)}_{ub}({\bm{V}}),
	\vspace{-2mm}  
\end{equation}
where ${\bm{V}}^{(r)}$ is a given point for the first-order Taylor approximation in the $r$th iteration and $\boldsymbol{u}_{\max }\left({\bm{V}}^{(r)}\right)$ denotes the eigenvector corresponding to the largest eigenvalue of matrix ${\bm{V}}^{(r)}$. As a result, the non-convex term $\|{\bm{V}}\|_{2}$ in the objective function of problem (P1-TD.3) can be approximated by ${\cal{U}}_{ub}^{(r)}({\bm{V}})$, and the approximated  problem can be updated and solved in an iterative manner by existing convex optimization solvers such as CVX \cite{Michael2014cvx}. Specifically, we propose a two-layer algorithm. In the outer layer, the value of the penalty coefficient $\eta$ is gradually reduced via $\eta^{(r+1)} \to e \eta^{(r)}$, where $e$ ($0 < e < 1$) is a scaling factor. In the inner layer, for the given penalty coefficient, the objective function is non-decreasing in each iteration. When $\frac{1}{\eta}$ increases to a sufficiently large value, equality constraint (\ref{RelaxingRankoneConstraintsA}) is ultimately satisfied.

\newtheorem{myDef}{\bf Definition}{
	\begin{myDef}
		(\textit{Maximum Sensing Frequency}): The maximum sensing frequency for a given target is given by $F_{s} = \frac{1}{K (\delta N_p + T_w)}$, where $\delta$ is the PRI and $T_w$ denotes the guard period between two consecutive sensing targets. Then, the dwell time spend on one target is $\delta N_p + T_w$.
	\end{myDef}
}

Intuitively, \textit{TD sensing} can asymptotically provide a squared-power gain in terms of the beam pattern gain in the desired direction \cite{wu2019beamformingDiscrete}. However, the sensing period, $K (\delta N_p + T_w)$, may be too large for time-sensitive tasks, especially, when the number of targets $K$ is large. In this case, it takes a long time to scan all the targets, and then the long time interval between two consecutive scans may result in outdated sensing information. This issue can be overcome by \textit{SS sensing}, which facilitates concurrent sensing of all targets.

\subsection{SS Sensing}
\label{SSMechanism}
\par  
The main challenges for solving problem (P1-SS) are the coupling between the beamforming vectors and the phase shifts in non-convex constraint (\ref{P1-SS-b}) and unit-modulus constraint (\ref{P1-SS-d}). To tackle this issue, we introduce auxiliary variables for decoupling the beamforming vectors and the phase shifts. Subsequently, problem (P1-SS) is solved by solving a series of simplified sub-problems. To this end, let ${\bm{W}}_{k} =  {\bm{w}}_{k}  {\bm{w}}^H_{k}$, with ${\rm{rank}}({\bm{W}}_{k}) = 1$ and ${\bm{W}}_{k} \succeq 0$, $\forall k$. Then, problem (P1-SS) can be equivalently transformed to
\vspace{-2mm}  
\begin{subequations}\label{P1-Co.1}
	\begin{align}
	(\rm{P1-SS.1}): \quad \label{P1-Co.1-a}& \begin{array}{*{20}{c}}
		\mathop {\max }\limits_{\{{\bm{W}}_{k}\}, \{{\bm{V}}\}, R} \quad \mathop R
	\end{array} \\ 
	\mbox{s.t.}\quad
	\label{P1-Co.1-b} & {\rm{Tr}}\left({{\bm{W}}_{k}} {\bm{Q}}^H_k {\bm{V}} {\bm{Q}}_k   \right) \ge R, \forall k \in {\cal{K}}, \\
	\label{P1-Co.1-c} & \sum\nolimits_{k'\in {\cal{K}} \backslash k }  {\rm{Tr}}\left({{\bm{W}}_{k}} {\bm{Q}}^H_{k'} {\bm{V}} {\bm{Q}}_{k'}  \right)  \le \varepsilon, \forall k \in {\cal{K}},  \\
	\label{P1-Co.1-d} &  \sum\nolimits_{k = 1} ^{K}  {\rm{Tr}} \left({{\bm{W}}_{k}}\right) \le P^{\max},  \\
	\label{P1-Co.1-e} & {\rm{rank}}({\bm{W}}_{k}) = 1, \forall k \in {\cal{K}},\\
	\label{P1-Co.1-f} & {\rm{rank}}({\bm{V}}) = 1,  \\
	\label{P1-Co.1-g} & {\bm{W}}_{k} \succeq 0, {\bm{V}} \succeq 0, \forall k \in {\cal{K}}. 
	\vspace{-2mm}  
\end{align} 
\end{subequations}
Solving problem (\rm{P1-SS.1}) is still challenging due to the non-convex rank-one constraints. To deal with this issue, we first relax the rank-one constraint of beamforming matrix ${\bm{W}}_k$, and denote the resulting problem by (P1-SS-SDR). Then, we prove that an optimal rank-one solution of problem (\rm{P1-SS.1}) can always be obtained based on the optimal solution of problem (P1-SS-SDR) without the rank-one constraints of matrices $\{{\bm{W}}_{k}\}$. The details are given as follows.
\begin{Pro}\label{RankOneCosensing}
	There always exists an optimal rank-one solution $\bar{\bm{W}}_{k}$ of problem (\rm{P1-SS.1}) that can be obtained based on the optimal solution $\{ \{{\bm{W}}^*_{k}\}, {\bm{V}}^*\}$ of problem (P1-SS-SDR), i.e.,
	\begin{equation}\label{ConstructionRank1Solution}
		\bar{\bm{w}}_{k} = ({\bm{v}}^{*H} {\bm{Q}}_k {\bm{W}}_{k} {\bm{Q}}^H_k {\bm{v}}^{*} )^{-1/2} {\bm{W}}^*_{k} {\bm{Q}}^H_k {\bm{v}}^{*}, \quad \bar{\bm{W}}_{k} = \bar{\bm{w}}_{k} \bar{\bm{w}}^H_{k}, \quad \bar{\bm{V}} = {\bm{V}}^*, \quad \forall k \in {\cal{K}}. 
	\end{equation}
\end{Pro}
\begin{proof}
	Please refer to Appendix A.
\end{proof} 

Based on Proposition \ref{RankOneCosensing}, the relaxed problem (P1-SS-SDR) without rank-one constraint (\ref{P1-Co.1-e}) is equivalent to the original problem ({P1-SS.1}), since the optimality of the solution is not compromised. However, the beamforming and phase shift vectors are still closely coupled in non-convex constraints ({\ref{P1-Co.1-b}}) and ({\ref{P1-Co.1-c}}). To address this issue, the inner approximation technique is adopted \cite{marks1978general}. To this end, the term ${\rm{Tr}}\left({{\bm{W}}_{k}}  {\bm{Q}}^H_k {\bm{V}} {\bm{Q}}_k   \right)$ in ({\ref{P1-Co.1-b}}) is rewritten as follows:
\vspace{-2mm}  
\begin{equation}\label{ContraintsCosensingPowerTarget}
	{\rm{Tr}}\left({{\bm{W}}_{k}}  {\bm{Q}}^H_k {\bm{V}} {\bm{Q}}_k   \right)  = \frac{1}{2}\left\| {{\bm{W}}_{k}}  +  {\bm{Q}}^H_k {\bm{V}} {\bm{Q}}_k  \right\|^2_F  - \frac{1}{2} \left\| {{\bm{W}}_{k}} \right\|_F^2  - \frac{1}{2} \left\|{\bm{Q}}^H_k {\bm{V}} {\bm{Q}}_k \right\|_F^2.
	\vspace{-2mm}  
\end{equation}
Note that the term $\frac{1}{2}\left\| {{\bm{W}}_{k}}  +  {\bm{Q}}^H_k {\bm{V}} {\bm{Q}}_k  \right\|^2_F$ in ({\ref{ContraintsCosensingPowerTarget}}) is non-concave with respect to ${{\bm{W}}_{k}}$ and ${\bm{V}}$. Therefore, we lower bound this term by adopting a first-order Taylor approximation. Specifically, in the $r$th iteration, we have
\vspace{-2mm}  
\begin{equation}\label{ExpressionTrace1}
	\begin{aligned}
		\left\| {{\bm{W}}_{k}}  +  {\bm{Q}}^H_k {\bm{V}} {\bm{Q}}_k  \right\|^2_F \ge & - \left\| {{\bm{W}}^{(r)}_{k}}  +  {\bm{Q}}^H_k {\bm{V}}^{(r)} {\bm{Q}}_k  \right\|^2_F \\  
		&+ 2{\rm{Re}}\left\{{\rm{Tr}}\left( {{\left( {{{{\bm{W}}_{k}}^{(r)}}  + {{\bm{Q}}^H_k {\bm{V}}^{(r)} {\bm{Q}}_k}} \right)}^H} \left({\bm{W}}_{k} + {\bm{Q}}^H_k {\bm{V}} {\bm{Q}}_k\right)\right)\right\}. 
		\vspace{-2mm}   
	\end{aligned}
\end{equation}
Then, a lower bound for the term in (\ref{ContraintsCosensingPowerTarget}) is given by
\begin{equation}\label{LowerBoundTrace}
	\begin{aligned}
		{\rm{Tr}}\left({{\bm{W}}_{k}}  {\bm{Q}}^H_k {\bm{V}} {\bm{Q}}_k   \right)  \ge& - \frac{1}{2}\left\| {{\bm{W}}^{(r)}_{k}}  +  {\bm{Q}}^H_k {\bm{V}}^{(r)} {\bm{Q}}_k  \right\|^2_F  - \frac{1}{2} \left\| {{\bm{W}}_{k}} \right\|_F^2  - \frac{1}{2} \left\|{\bm{Q}}^H_k {\bm{V}} {\bm{Q}}_k \right\|_F^2 \\
		& + {\rm{Re}}\left\{{\rm{Tr}}\left( {{{\left( {{{{\bm{W}}^{(r)}_{k}}} + {{\bm{Q}}^H_k {\bm{V}}^{(r)} {\bm{Q}}_k}} \right)}^H}{ \left( {\bm{W}}_{k} + {\bm{Q}}^H_k {\bm{V}} {\bm{Q}}_k \right) } } \right)\right\} \\
		& \triangleq {{\cal{F}}^{(r)}_{lb}\left({{\bm{W}}_{k}}  {\bm{Q}}^H_k {\bm{V}} {\bm{Q}}_k   \right)}.
	\end{aligned}
\end{equation}
Similarly, for constraint (\ref{P1-Co.1-c}), for the given points ${{\bm{W}}_{k}^{(r)}}$ and ${\bm{V}}^{(r)}$, we have
\begin{equation}
	-\left\| {{\bm{W}}_{k}} \right\|_F^2 \le \left\|{{\bm{W}}_{k}^{(r)}}\right\|_{F}^{2} - 2 {\rm{Re}}\left\{ \operatorname{Tr}\left(\left({{\bm{W}}^{(r)}_{k}}\right)^{H} {{\bm{W}}_{k}}\right)\right\},
\end{equation}
\begin{equation}
	- \left\|{\bm{Q}}^H_k {\bm{V}} {\bm{Q}}_k \right\|_F^2 \le \left\|{\bm{Q}}^{H}_{k} {\bm{V}}^{(r)} {\bm{Q}}_{k}\right\|_{F}^{2} - 2 {\rm{Re}}\left\{\operatorname{Tr}\left(\left({\bm{Q}}_{k} {\bm{Q}}^{H}_{k} {\bm{V}}^{(r)} {\bm{Q}}_{k} {\bm{Q}}^{H}_{k}\right)^{H} {\bm{V}}\right)\right\}.
\end{equation}
Then, the term $\operatorname{Tr}\left({{\bm{W}}_{k}} {\bm{Q}}^H_{k'} {\bm{V}} {\bm{Q}}_{k'}\right)$ in constraint (\ref{P1-Co.1-c}) satisfies
\begin{equation}\label{UpperBoundTrace}
	\begin{aligned}
		\operatorname{Tr}\left({{\bm{W}}_{k}} {\bm{Q}}^H_{k'} {\bm{V}} {\bm{Q}}_{k'}\right) \le &   \frac{1}{2}\left\| {{\bm{W}}_{k}}  +  {\bm{Q}}^H_{k'} {\bm{V}} {\bm{Q}}_{k'}  \right\|^2_F  + \frac{1}{2}\left\|{{\bm{W}}^{(r)}_{k}}\right\|_{F}^{2}-{\rm{Re}}\left\{ \operatorname{Tr}\left(\left({{\bm{W}}^{(r)}_{k}}\right)^{H} {{\bm{W}}_{k}}\right)\right\} \\
		&   + \frac{1}{2}\left\|{\bm{Q}}_{k'}^{H} {\bm{V}}^{(r)} {\bm{Q}}_{k'}\right\|_{F}^{2} - {\rm{Re}}\left\{\operatorname{Tr}\left(\left({\bm{Q}}_{k'} {\bm{Q}}^{H}_{k'} {\bm{V}}^{(r)} {\bm{Q}}_{k'} {\bm{Q}}^{H}_{k'}\right)^{H} {\bm{V}}\right)\right\} \\ 
		& \triangleq {\cal{F}}^{(r)}_{ub}\left({{\bm{W}}_{k}} {\bm{Q}}^H_{k'} {\bm{V}} {\bm{Q}}_{k'}\right).
	\end{aligned}
\end{equation}
After respectively replacing the terms ${\rm{Tr}}\left({{\bm{W}}_{k}}  {\bm{Q}}^H_k {\bm{V}} {\bm{Q}}_k   \right)$ and $\operatorname{Tr}\left({{\bm{W}}_{k}} {\bm{Q}}^H_{k'} {\bm{V}} {\bm{Q}}_{k'}\right)$ in constraints (\ref{P1-Co.1-b}) and (\ref{P1-Co.1-c}) by those in (\ref{LowerBoundTrace}) and (\ref{UpperBoundTrace}), a similar penalty-based approach as in (\ref{RelaxingRankoneConstraintsA})-(\ref{RelaxingRankoneConstraintsB}) is adopted to handle the non-convex rank-one constraint of ${\bm{V}}$ in ({\ref{P1-Co.1-f}}). Then, for any local points ${\bm{W}}_{k,l}^{(r)}$ and ${\bm{V}}^{(r)}$ in the $r$th iteration, the beamforming and phase shift matrices are optimized as follows:
\begin{subequations}\label{P1-Co.2}
\begin{align}
	(\rm{P1-SS.2}): \quad \label{P1-Co.2-a}&	 \mathop {\max }\limits_{\{{\bm{W}}_{k}\}, \{{\bm{V}}\}, R} \quad \mathop R - \frac{1}{\eta}{\cal{U}}^{(r)}_{ub}({\bm{V}}) \\ 
	\mbox{s.t.}\quad
	& (\rm{\ref{P1-Co.1-d}}), \nonumber \\
	\label{P1-Co.2-b}&  {{\cal{F}}^{(r)}_{lb}\left({{\bm{W}}_{k}}  {\bm{Q}}^H_k {\bm{V}} {\bm{Q}}_k   \right)} \ge R, \forall k \in {\cal{K}}, \\
	\label{P1-Co.2-c}& \sum\nolimits_{k'\in {\cal{K}} \backslash k }  {\cal{F}}^{(r)}_{ub}\left({{\bm{W}}_{k}} {\bm{Q}}^H_{k'} {\bm{V}} {\bm{Q}}_{k'}\right)  \le \varepsilon, \forall k \in {\cal{K}}, \\
	\label{P1-Co.2-d}& {\bm{W}}_{k} \succeq 0, {\bm{V}} \succeq 0, \forall k \in {\cal{K}}.
\end{align} 
\end{subequations}

The above optimization problem can be solved in the $r$th iteration of the inner approximation procedure by convex optimization solvers such as CVX \cite{Michael2014cvx}. The objective function in (\ref{P1-Co.2}) is non-decreasing in each iteration since the local points ${\bm{W}}_{k,l}^{(r)}$ and ${\bm{V}}^{(r)}$ are feasible for problem (\rm{P1-SS.2}), and its corresponding objective value is upper bounded due to the limited transmit power $P^{\max}$. Thus, the proposed algorithm is guaranteed to converge to a locally optimal KKT solution of problem (P1-SS.1) according to \cite{tao1997convex}.

To further explore IRS-assisted SS sensing, we analyze the special case of two targets in the next subsection, where unlike for the general case, the optimal BS beamforming vector for problem (P1-SS) can be found in closed form with low complexity.

\subsection{Two-Target SS Sensing}
For $K=2$ targets,\footnote{The two-target case is relevant in some practical applications, e.g., for vehicle detection, only the entrance and the exit of the road need to be monitored in real-time; for target tracking, the tracker and the object to be tracked are typically jointly sensed.} let matrix ${\bm{H}}_k^{\bot}$ be the subspace orthogonal to an arbitrary vector ${\bm{h}}_k$, i.e., ${{\bm{h}}_k^{H}} {\bm{H}}_k^{\bot} = {\bm{0}}$ and ${\bm{H}}_k^{\bot}$ is given by 
\vspace{-2mm} 
\begin{equation}\label{NullSpace}
	{\bm{H}}_k^{\bot}=\left(\bm{I}_{M}-{\bm{h}}_k\left({\bm{h}}^{H}_k {\bm{h}}_k\right)^{-1} {\bm{h}}^{H}_k\right) \in \mathbb{C}^{M \times M}.
	\vspace{-2mm} 
\end{equation}
Then, let ${\bm{h}}_{k,k'}^{\bot} = {\bm{H}}_{k'}^{\bot} {\bm{h}}_k = {\bm{h}}_k - \frac{{\bm{h}}^{H}_{k'} {\bm{h}}_k}{{\bm{h}}^{H}_{k'} {\bm{h}}_{k'}}{\bm{h}}_{k'} \in \mathbb{C}^{M \times 1}$, i.e., ${\bm{h}}_{k,k'}^{\bot}$ is orthogonal to vector ${\bm{h}}_{k'}$, and ${\bm{h}}_{k,k'}^{\bot}$ lies in the space spanned by vectors ${\bm{h}}_k$ and ${\bm{h}}_{k'}$. 

\begin{Pro}\label{TargetSensingStructure}
	When $K = 2$, for given phase shift matrix ${\bm{\Theta}}$ and transmit power $p_k = \|{\bm{w}}_k\|^2$, the optimal beamforming vector of problem (P1-SS) can be expressed as 
	\vspace{-2mm} 
	\begin{equation}\label{BeamformingStructure}
		{\bm{w}}_k^* = \left\{ {\begin{array}{*{20}{c}}
				{\sqrt{p_k}\frac{{\bm{h}}_{k}}{\|{\bm{h}}_{k}\|},}&{ \frac{{{{p_{k}}|{\bm{h}}^H_{k'}{\bm{h}}_{k}|^2}}}{\|{\bm{h}}^H_{k}\|^2} \le  \varepsilon}\\
				{\rho_1  {\bm{h}}_k  + \rho_2 \frac{{\bm{h}}_{k,k'}^{\bot}}{\|{\bm{h}}_{k,k'}^{\bot}\|} {\cos \psi _{h_{k'}^{\bot},h_k}},}&{{\rm{otherwise}}}
		\end{array}} \right.,  k, k' \in \{1,2\},
	\vspace{-2mm} 
	\end{equation}
	where ${\bm{h}}_k = {\bm{G}}^H {\bm{\Theta}}^H_l {\bm{a}}_k $, ${\bm{h}}_{k,k'}^{\bot} = {\bm{H}}_{k'}^{\bot} {\bm{h}}_k$, $\arccos \angle \psi_{H_{k'}^{\bot},h_k} = \frac{\left({\bm{h}}^H_k{\bm{h}}_{k,k'}^{\bot}\right)^H}{|{\bm{h}}^H_k{\bm{h}}_{k,k'}^{\bot}|}$, $\rho_1 = \frac{\sqrt{\varepsilon}}{|{\bm{h}}^H_{k'} {\bm{h}}_k|}$, and $\rho_2 = \frac{{ - \rho_1\left| {{\bm{h}}_k^H{\bm{h}}_{k,k'}^{\bot} } \right| + \sqrt {{{\rho_1^2}}{{\left| {{\bm{h}}_k^H{\bm{h}}_{k,k'}^{\bot} } \right|}^2} - {{\left\| {{\bm{h}}_{k,k'}^{\bot} } \right\|}^2}\left( {{\rho_1^2}{{\left\| {{\bm{h}}_k^H} \right\|}^2} - p_k} \right)} }}{{{{\left\| {{\bm{h}}_{k,k'}^{\bot} } \right\|}^2}}}$.
\end{Pro}   
\begin{proof}
	Please refer to Appendix B.
\end{proof}

Proposition \ref{TargetSensingStructure} shows that the optimal transmit beamforming vector lies in the space spanned by effective channel vectors ${\bm{h}}_k$ and ${\bm{h}}_{k,k'}^{\bot}$. In particular, when $\varepsilon = 0$, the optimal beamforming vector simplifies to ${\bm{w}}_k^* = \sqrt{p_k} \frac{{\bm{h}}_{k,k'}^{\bot}}{\|{\bm{h}}_{k,k'}^{\bot}\|}$.

By plugging the optimal beamforming vector from Proposition \ref{TargetSensingStructure} from problem (P1-SS), the beam pattern gain for direction $k$ is given by ${\cal{P}}_k = \frac{{\| {{\bm{h}}_k^H} \|^2}}{{\| {{\bm{h}}_{k'}} \|^2}}\left( \Xi{\sqrt {\varepsilon }  {{ +  \sqrt{1-\Xi^2}}}\sqrt {{p_k}{{\| {\bm{h}}_{k'} \|}^2} - \varepsilon } } \right)^2$ when $ \frac{{{{p_{k}}|{\bm{h}}^H_{k'}{\bm{h}}_{k}|^2}}}{\|{\bm{h}}^H_{k}\|^2} >  \varepsilon$, where $k \in \{1,2\}$ and $\Xi = \frac{{\bm{h}}^H_k{\bm{h}}_{k'}}{\|{\bm{h}}^H_k\|\|{\bm{h}}_{k'}\|}$; otherwise, ${\cal{P}}_k = {p_k} \left\| {\bm{h}}_k \right\|^2$. Intuitively, the beam pattern gain of each target direction increases monotonically with the transmit power of the corresponding beam, and thus, for the two-target case, at the optimal solution of problem (P1-SS), ${\cal{P}}^*_1 = {\cal{P}}^*_2$. Due to the closely coupled variables, we adopt the alternating optimization (AO) technique to solve this problem. Specifically, in the $r$th iteration, for given phase shift vector ${\bm{v}}^{(r)}$, the optimal beam power $p^{(r+1)}_k$ can be obtained with a 1-D binary search based on Proposition \ref{TargetSensingStructure}, thereby yielding the optimal beamforming vector for a given phase shift vector. For given beamforming vector $\{{\bm{w}}^{(r)}_k\}$, the phase shift optimization problem can be simplified as 
\begin{subequations}\label{P1-CoSub2}
\begin{align}
	(\rm{P1-SS.3}): \quad \label{P1-CoSub2-a} &	\mathop {\max }\limits_{{\bm{V}}} \quad \mathop R \\ 
	\mbox{s.t.}\quad
	\label{P1-CoSub2-b}& {\rm{Tr}} \left({\bm{Q}}_k {{\bm{w}}^*_{k}} {{\bm{w}}^{*H}_{k}} {\bm{Q}}^H_k  {\bm{V}}\right) \ge R, \forall k \in {\cal{K}}, \\
	\label{P1-CoSub2-c}& {\rm{Tr}} \left({\bm{Q}}_{k'} {{\bm{w}}^*_{k}} {{\bm{w}}^{*H}_{k}} {\bm{Q}}^H_{k'}  {\bm{V}}\right) \le \varepsilon, \forall k, k' \in {\cal{K}}, \\
	\label{P1-CoSub2-d}& |[{\bm{V}}]_{n,n}| = 1, \forall n \in {\cal{N}}, \\
	\label{P1-CoSub2-e}	& {\rm{rank}}\left( {\bm{V}} \right) = 1, {\bm{V}} \succeq 0. 
\end{align} 
\end{subequations}
By adopting a similar relaxation method as in Section \ref{TDSensingMechanism} (c.f. (\ref{RelaxingRankoneConstraintsA})-(\ref{RelaxingRankoneConstraintsB})), problem (\ref{P1-CoSub2}) can be solved by existing convex optimization solvers. The details are omitted considering the page limitation. Then, a sub-optimal solution of problem (\rm{P1-SS}) for the two-target case can be obtained by determining the optimal beamforming vector ${\bm{w}}_k$ and the local-optimal reflection-coefficient vector ${\bm{v}}$ in an alternating manner. The details of the proposed AO algorithm are provided in {\bf{Algorithm} \ref{AOAlgorithm}}.

\begin{algorithm}[t]
	\small
	\caption{AO Algorithm}
	\label{AOAlgorithm}
	\begin{algorithmic}[1]
		\STATE {\bf{Initialize}}  $\{{\bm{w}}_{k}\}$, $\{p_k\}$, and ${\bm{v}}$, iteration number $r = 1$, convergence accuracy $\epsilon_1$, objective value $V^{(r)*}$
		\REPEAT 
		\STATE For given ${\bm{v}}^{(r)}$, obtain $\{p_k^{(r+1)}\}$ by a 1-D binary search based on the obtained structure of the optimal beamforming vectors $\{{\bm{w}}_k^*\}$ in Proposition \ref{TargetSensingStructure}
		\STATE For given ${\bm{v}}^{(r)}$ and $\{p_k^{(r+1)}\}$, obtain $\{{\bm{w}}_k^{(r+1)}\}$ according to Proposition \ref{TargetSensingStructure}
		\STATE For given $\{{\bm{w}}_k^{(r+1)}\}$, obtain ${\bm{v}}^{(r+1)}$ by solving problem (\rm{P1-SS.3})
		\STATE Update the objective value $V^{(r+1)*}$ of problem (P1-SS.1) according to the obtained $\{{\bm{w}}^{(r+1)}_{k}\}$, $\{p_k^{(r+1)}\}$, and ${\bm{v}}^{(r+1)}$
		\STATE $r = r + 1$
		\UNTIL $\left|V^{(r)*} - V^{(r-1)*}\right| \le \epsilon_1$
	\end{algorithmic}
\end{algorithm}

Intuitively, for SS sensing, the beam pattern gain for each target decreases monotonically as the number of targets increases, since the total power is divided between different directions. To enable a flexible trade-off between the beam pattern gain and the sensing frequency, in the next section, we propose a more general IRS-based sensing scheme with controllable sensing power and sensing period.

\section{Proposed Hybrid Sensing Scheme}
\label{AnalysisGeneral}
In this section, we propose a \textit{hybrid TD-SS sensing} scheme, where the targets in the same group are sensed simultaneously via \textit{SS sensing} while targets in different groups are assigned to orthogonal time slots. By adjusting the number of groups, the proposed hybrid scheme can provide a flexible balance between beam pattern gain and sensing frequency. In practice, the number of groups can be set and updated based on the required minimum sensing frequency and measurements of the sensing performance. The considered system model and problem formulation are presented in Section \ref{HybridModel}, followed by the proposed penalty-based algorithm to jointly optimize the transmit beamforming vectors, IRS phase shifts, and target grouping in Section \ref{PenaltyAlgorithm}.

\subsection{Hybrid Sensing Model and Problem Formulation}
\label{HybridModel}
\begin{figure*}[t]
	\centering
	\setlength{\abovecaptionskip}{0.cm}
	\subfigure[Transmited and received signals over time.]
	{	
		\label{figure3a}
		\includegraphics[width=7.8cm]{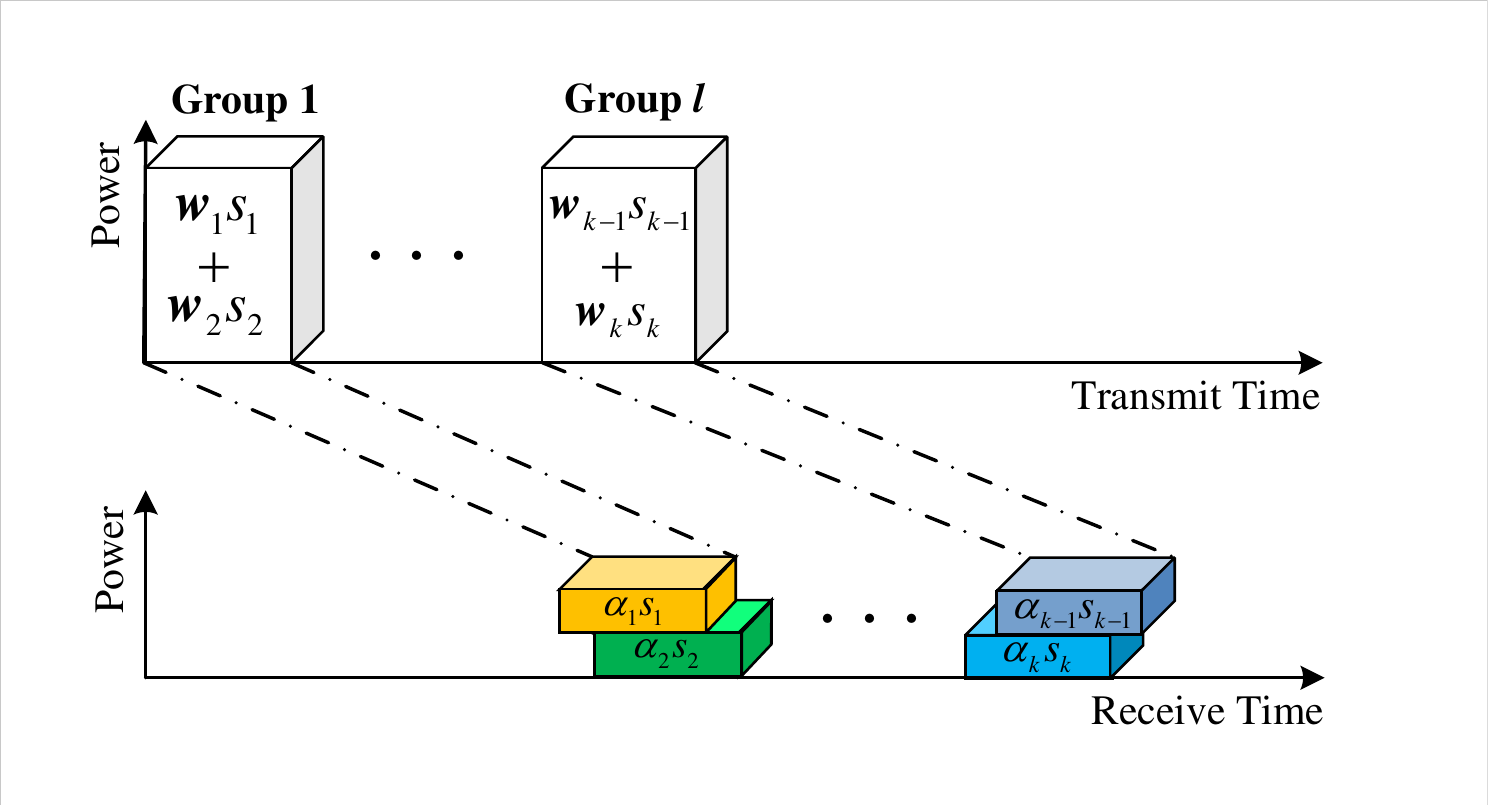}
	}
	\subfigure[Equavilent received signals from target directions.]
	{	
		\label{figure3b}
		\includegraphics[width=7.1cm]{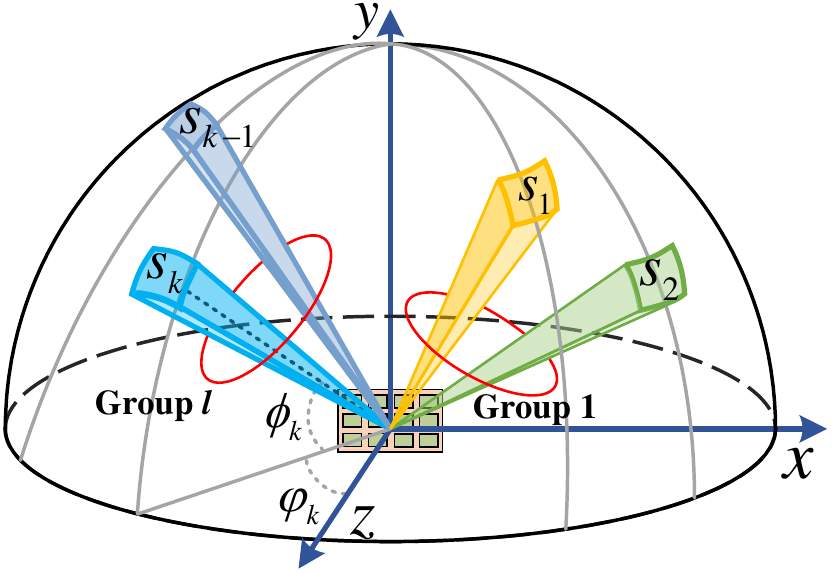}
	}	
	\caption{Illustration of the proposed \textit{hybrid TD-SS sensing} scheme for the case with $K/L = 2$.}
	\label{figure3}
\end{figure*}
To achieve a flexible trade-off between beam pattern gain and sensing frequency, the $K$ sensing directions/targets are partitioned into $L$ groups and the set of potential targets in the $l$th group is denoted by ${\cal{K}}_l$, $l \in {\cal{L}} =  \{1,\cdots,L\}$. Accordingly, the targets in different groups are sensed in a time-division manner as in \textit{TD sensing}, while the targets in the same group are sensed simultaneously as in \textit{SS sensing}. The \textit{hybrid TD-SS sensing} scheme is illustrated in Fig.~\ref{figure3} for $K/L = 2$, where two dedicated radar/sensing beams are jointly transmitted in each orthogonal time slot, and the reflected signals belonging to different groups do not interfere with each other. Two dedicated beams modulated with SSs $\{{\bm{s}}_{i},{\bm{s}}_{j}\}$ are activated in the $l$th time slot, i.e., $i,j \in {\cal{K}}_l$, where \textit{SS sensing} is adopted, as shown in Fig.~\ref{figure3b}. Then, the duration of a detection epoch is given by $L (\delta N_p + T_w)$, which is $L/K$ times the duration of a detection epoch for \textit{TD sensing}. If the BS senses the target in direction $k$ in the $l$th  group, we let $c_{k,l} = 1$, otherwise, $c_{k,l} = 0$. Also, each direction should be sensed once during the entire sensing epoch. Based on the above discussion, the following condition holds:
\begin{equation}\label{SensingTimeConditionA}
	\sum\nolimits_{l = 1}^{L} {{c_{k,l}}} = 1, \forall k \in {\cal{K}}. 
\end{equation}
For the $l$th sensing group, at time $t$, the complex baseband transmitted signal at the BS can be expressed as follows:
\begin{equation} 
	{\bm{x}}_l[t] = \sum\nolimits_{k \in {\cal{K}}_l}{{{\bm{w}}_{k,l}} {s_{k}}}[t],
\end{equation}
where ${{\bm{w}}_{k,l}} \in \mathbb{C}^{M \times 1}$ is the beamforming vector for target $k$ in group $l$. For sensing the targets in group $l$, the reflection-coefficient matrix of the IRS is denoted by ${\bm{\Theta}}_l = {\rm{diag}}(e^{j \theta_{l,1}}, ... , e^{j \theta_{l,N}})$. As a result, during the dwell time of sensing target $k$ for SS ${\bm{s}}_k$, the beam pattern gain of the IRS in direction $\{\varphi_k, \phi_k\}$ is given by
\begin{equation}
	\mathcal{P}_k  = \sum\nolimits_{l = 1}^L c_{k,l} {\bm{v}}^H_l {\bm{Q}}_k {{\bm{w}}_{k,l}}{{\bm{w}}^H_{k,l}} {\bm{Q}}^H_k {\bm{v}}_l,
\end{equation}
where ${\bm{v}}^H_l = [e^{-j\theta_{l,1}},\cdots,e^{-j\theta_{l,N}}]$. For the proposed \textit{hybrid TD-SS sensing} scheme, the transmit beamformer at the BS, the IRS phase shifts, and the target grouping are jointly optimized for maximization of the beam pattern gain, i.e.,
\begin{subequations}\label{P2}
\begin{align}	
	(\rm{P2}): \quad \label{P2-a}&	\mathop {\max }\limits_{\{{\bm{w}}_{k,l}\}, \{{\bm{v}}_l\}, \{c_{k,l}\}} \quad \mathop {\min }\limits_{k \in {\cal{K}}} \quad \mathop  {{\mathcal{P}}_k} \\ 
	\mbox{s.t.}\quad
	\label{P2-b}& \sum\nolimits_{k'\in {\cal{K}} \backslash k } \sum\nolimits_{l = 1}^L c_{k',l} c_{k,l}  {\bm{v}}^H_l {\bm{Q}}_{k'} {\bm{w}}_{k,l}  {\bm{w}}^H_{k,l} {\bm{Q}}^H_{k'} {\bm{v}}_l \le \varepsilon, \forall k \in {\cal{K}},   \\
	\label{P2-c}&  \sum\nolimits_{k = 1} ^{K} c_{k,l}\| {\bm{w}}_{k,l} \|^2 \le P^{\max},  \forall l \in {\cal{L}}, \\
	\label{P2-d}& |[{\bm{v}}_l]_{n}| = 1, \forall n \in {\cal{N}}, l \in {\cal{L}},  \\
	\label{P2-e}& \sum\nolimits_{l = 1}^{L} {{c_{k,l}}} = 1, \forall k \in {\cal{K}},  \\
	\label{P2-f}& c_{k,l} \in \{0,1\},  \forall k \in {\cal{K}}, l \in {\cal{L}}. 
\end{align} 
\end{subequations}
In problem (\rm{P2}), (\ref{P2-b}) specifies that the power leakage of the beam intended for target $k$ in other directions should be lower than the tolerable threshold $\varepsilon$. The total transmit power and the reflection-coefficient vector are constrained by (\ref{P2-c}) and (\ref{P2-d}), respectively. Constraint (\ref{P2-e}) ensures that each target is sensed once during one sensing epoch. To gain more insights, a useful and fundamental trade-off between beam pattern gain and sensing frequency is rigorously proved in Theorem \ref{SensingPowerVersusFrequency}. 
\begin{theorem}\label{SensingPowerVersusFrequency}
	For $L$ groups, the optimal beam pattern gain in problem (\rm{P2}) $\mathop {\min }\limits_{k \in {\cal{K}}}  \mathop  {{\mathcal{P}}^*_k} (L)$ meets the  inequality $\mathop {\min }\limits_{k \in {\cal{K}}} \mathop  {{\mathcal{P}}^*_k} (L+1)  \ge \mathop {\min }\limits_{k \in {\cal{K}}} \mathop  {{\mathcal{P}}^*_k}(L)$.
\end{theorem}
\begin{proof}
	The optimal beamforming vectors, phase shifts, and target groupings for $L$ groups are denoted by ${\bm{X}}^* = \{\{{\bm{w}}^*_{k,l}\}, \{{\bm{v}}^*_l\}, \{c^*_{k,l}\}\}$, respectively, and the corresponding maximum beam pattern gain in (\ref{P2-a}) is denoted by $\mathop {\min }\limits_{k \in {\cal{K}}} \mathop  {{\mathcal{P}}^*_k}(L)$. If an empty group without target is added, indexed by $L+1$, it can be readily proved that the optimal solution ${\bm{X}}^*$ for $L$ groups also satisfies constraints (\ref{P2-b})-(\ref{P2-f}) for the newly constructed $L+1$ groups of targets.
	Based on task allocation $\{c^*_{k,l}\}$, another feasible solution $\{\hat c_{k,l}\}$ can be constructed as follows. Let $l' = \arg \mathop {\min }\limits_l \mathop {\min }\limits_{k \in {{\cal{K}}_l}} {{\mathcal{P}}_k}$, which represents the group containing the target with the minimum beam pattern gain among all targets. A certain target $\bar k$ in group $l'$ can be moved to group $L+1$ while satisfying constraints (\ref{P2-e})-(\ref{P2-f}), i.e., 
	\begin{equation}
	 \hat c_{\bar k,l} = 0, \quad \hat c_{\bar k,L+1} = 1, \quad \text{and} \quad \hat c_{k,l} = c^*_{k,l}, \forall k \ne \bar k.
	\end{equation}
 	In this case, with the optimal solution $\{{\bm{w}}^*_{\bar k,l}\}, \{{\bm{v}}^*_l\}$, the beam pattern gain of target $\bar k$ is ${{\mathcal{P}}^*_{\bar k}}$, which satisfies ${{\mathcal{P}}^*_{\bar k}} \ge \mathop {\min }\limits_{k \in {\cal{K}}} \mathop  {{\mathcal{P}}^*_k}(L)$. Also, the interference constraints in group $l'$ become looser, i.e.
 	\begin{equation}
 		\begin{aligned}
 			&\sum\nolimits_{k'\in {\cal{K}}_{l'} \backslash \{k, \bar k\} } \sum\nolimits_{l = 1}^L \bar c_{k',l} \bar c_{k,l}  {\bm{v}}^{*H}_l {\bm{Q}}^*_{k'} {\bm{w}}^*_{k,l}  {\bm{w}}^{*H}_{k,l} {\bm{Q}}^H_{k'} {\bm{v}}^*_l \\
 			\le& \sum\nolimits_{k'\in {\cal{K}}_{l'} \backslash k } \sum\nolimits_{l = 1}^L \bar c_{k',l} \bar c_{k,l}  {\bm{v}}^{*H}_l {\bm{Q}}_{k'} {\bm{w}}^{*}_{k,l}  {\bm{w}}^{*H}_{k,l} {\bm{Q}}^H_{k'} {\bm{v}}^{*}_l \le \varepsilon
 		\end{aligned}
 	\end{equation}
     Also, more transmit power is available to further improve the beam pattern of the other targets in group $l'$. Hence, the optimal beam pattern gain under the constructed solution $\{\hat c_{k,l}\}$ is not smaller than $\mathop {\min }\limits_{k \in {\cal{K}}} \mathop  {{\mathcal{P}}^*_k} (L)$. On the other hand, it is not difficult to verify that the optimal beam pattern gain for \textit{TD sensing} is an upper bound for that for \textit{hybrid TD-SS sensing}, i.e., the optimal beam pattern gain for such a hybrid scheme cannot be larger than that for $L = K$. Based on the above analysis, the optimal beam pattern gain in (\ref{P2-a}) is monotonically non-increasing as the number of groups decreases, i.e., $\mathop {\min }\limits_{k \in {\cal{K}}} \mathop  {{\mathcal{P}}^*_k} (L+1) \ge \mathop {\min }\limits_{k \in {\cal{K}}} \mathop  {{\mathcal{P}}^*_k}(L)$.
\end{proof}

According to Theorem \ref{SensingPowerVersusFrequency}, since the sensing frequency $F_s$ is a non-increasing function with respect to the number of groups $L$ ($F_s = {1 \mathord{\left/	{\vphantom {1 {L{ \delta N_p}}}} \right.	\kern-\nulldelimiterspace} {L{ (\delta N_p + T_w)}}}$), the optimal beam pattern gain in problem (\rm{P2}) is a monotonically non-increasing function with respect to the sensing frequency $F_s$. Therefore, Theorem \ref{SensingPowerVersusFrequency} reveals the monotonic relationship between the sensing frequency and the beam pattern gain for the proposed \textit{hybrid TD-SS sensing} scheme. Intuitively, when the number of targets in a group increases, the maximum beam pattern gain decreases since the transmit power has to be divided into more directions while the interference constraints will become more stringent. Solving problem (\rm{P2}) is non-trivial due to the integer variables which are closely coupled with the beamforming vector and the IRS phase shifts. To handle this issue, by decomposing the constraints and decoupling the integer variables, we propose a two-layer penalty-based optimization algorithm, whose details are provided next.

\subsection{Penalty-based Algorithm for Solving {\rm{(P2)}}}
\label{PenaltyAlgorithm}
Due to the closely coupled variables and non-convex constraints, it is challenging to optimally solve problem (\rm{P2}). To tackle the non-convexity of constraint (\ref{P2-d}), we first relax it into a convex form and then recover ${\bm{v}}_l$ accordingly, i.e., \footnote{We note that large amplification factors are generally beneficial for maximizing the beam pattern gain. Hence, the optimal solution tends to yield $|[v_l]_n| = 1$ even for the relaxed constraint $|[v_l]_n| \le 1$.}
\vspace{-2mm}  
\begin{equation}\label{UnitModule}
	|[{\bm{v}}_l]_{n}| \le 1, \forall n \in {\cal{N}}, l \in {\cal{L}}.
	\vspace{-2mm}  
\end{equation}
To facilitate the beamforming and phase shift design, problem (\rm{P2}) is rewritten as follows:
\vspace{-2mm}  
\begin{subequations}\label{P2.1}
\begin{align}		
	({\rm{P2.1}}): \quad  \label{P2.1-a} &	\mathop {\max }\limits_{\{{\bm{W}}_{k,l}\}, \{{\bm{V}}_l\}, \{c_{k,l}\}, R} \quad \mathop R  \\ 
	\mbox{s.t.}\quad
	& (\rm{\ref{P2-e}}),(\rm{\ref{P2-f}}),(\ref{UnitModule}), \nonumber \\
	\label{P2.1-b}& \sum\nolimits_{l = 1}^L c_{k,l} {\bm{v}}^H_l {\bm{Q}}_k {{\bm{W}}_{k,l}} {\bm{Q}}^H_k {\bm{v}}_l \ge R, \forall k \in {\cal{K}},  \\
	\label{P2.1-c}& \sum\nolimits_{k'\in {\cal{K}} \backslash k } \sum\nolimits_{l = 1}^L c_{k',l} c_{k,l}  {\bm{v}}^H_l {\bm{Q}}_{k'} {\bm{W}}_{k,l} {\bm{Q}}^H_{k'} {\bm{v}}_l  \le \varepsilon, k \in {\cal{K}},  \\
	\label{P2.1-d}&  \sum\nolimits_{k = 1} ^{K} c_{k,l} {\rm{Tr}} \left({\bm{W}}_{k,l}\right) \le P^{\max},  \forall l \in {\cal{L}},  \\
	\label{P2.1-e}& {\rm{rank}}({\bm{W}}_{k,l}) = 1, \forall k \in {\cal{K}}, l \in {\cal{L}}, \\
	\label{P2.1-f}& {\bm{W}}_{k,l} \succeq 0, \forall k \in {\cal{K}}, l \in {\cal{L}}, 
	\vspace{-2mm}  
\end{align} 
\end{subequations}
where ${\bm{W}}_{k,l} =  {\bm{w}}_{k,l}  {\bm{w}}^H_{k,l}$. Let ${\tilde{\bm{W}}_{k,l}} = c_{k,l} {{\bm{W}}_{k,l}}$ and $\hat{\bm{W}}_{k,k',l} = c_{k',l} {\tilde{\bm{W}}_{k,l}}$. To further decompose the product term $c_{k',l} {\tilde{\bm{W}}_{k,l}}$, the big-M formulation is adopted \cite{griva2009linear, Sun2018Robust} by introducing the following convex constraints:
\vspace{-2mm}  
\begin{equation}\label{StrackConstraintsA}
	{\tilde{\bm{W}}_{k,l}} \preceq c_{k,l} P^{\max} {\bm{I}}_M; \quad {\tilde{\bm{W}}_{k,l}} \succeq 0, 
\end{equation}
\begin{equation}\label{StrackConstraintsB}
	\hat{\bm{W}}_{k,k',l} \preceq c_{k',l} P^{\max} {\bm{I}}_M; \quad {\hat{\bm{W}}_{k,k',l}} \succeq 0 
\end{equation}
\begin{equation}\label{StrackConstraintsC}
	{\hat{\bm{W}}_{k,k',l}} \preceq \tilde{\bm{W}}_{k,l}; \quad \hat{\bm{W}}_{k,k',l} \succeq \tilde{\bm{W}}_{k,l} - (1-c_{k',l})P^{\max} {\bm{I}}_{M}.
	\vspace{-2mm}  
\end{equation}
It can be readily proved that ${{\bm{W}}_{k,l}} = {\tilde{\bm{W}}_{k,l}}$ if $c_{k,l} = 1$, and ${\hat{\bm{W}}_{k,k',l}} = {\tilde{\bm{W}}_{k,l}}$ if $c_{k',l} = 1$; while ${{\bm{W}}_{k,l}} = {\tilde{\bm{W}}_{k,l}} = {\bm{0}}$ if $c_{k,l} = 0$, and ${\hat{\bm{W}}_{k,k',l}} = {\bm{0}}$ if $c_{k',l} = 0$. Hence, beamforming matrices $\{{\bm{W}}_{k,l}\}$ and target grouping variables $\{c_{k,l}\}$ can be recovered from the solution $\{c_{k,l}, {\tilde{\bm{W}}_{k,l}}, {\hat{\bm{W}}_{k,k',l}}\}$ with guaranteed uniqueness. Thus, problem (\ref{P2.1}) can be transformed as follows:
\vspace{-2mm}  
\begin{subequations}\label{P2.2}
\begin{align}
	({\rm{P2.2}}): \quad \label{P2.2-a} &	\mathop {\max }\limits_{\{\tilde{\bm{W}}_{k,l}\}, \{\hat{\bm{W}}_{k,k',l}\}, \{{\bm{v}}_l\}, \{c_{k,l}\}, R} \quad \mathop R \\ 
	\mbox{s.t.}\quad
	& (\rm{\ref{P2-e}}), (\rm{\ref{P2-f}}), (\ref{UnitModule}), (\ref{StrackConstraintsA})-(\ref{StrackConstraintsC}), \nonumber \\
	\label{P2.2-b}& \sum\nolimits_{l = 1}^L  {\bm{v}}^H_l {\bm{Q}}_k {\tilde{\bm{W}}_{k,l}} {\bm{Q}}^H_k {\bm{v}}_l \ge R, \forall k \in {\cal{K}}, \\
	\label{P2.2-c}& \sum\nolimits_{k'\in {\cal{K}} \backslash k } \sum\nolimits_{l = 1}^L  {\bm{v}}^H_l {\bm{Q}}_{k'} \hat{\bm{W}}_{k,l} {\bm{Q}}^H_{k'} {\bm{v}}_l  \le \varepsilon, \forall k \in {\cal{K}},  \\
	\label{P2.2-d}&  \sum\nolimits_{k = 1} ^{K}  {\rm{Tr}} \left({\tilde{\bm{W}}_{k,l}}\right) \le P^{\max},  \forall l \in {\cal{L}},\\
	\label{P2.2-e}& {\rm{rank}}(\tilde{\bm{W}}_{k,l}) = 1, {\rm{rank}}(\hat{\bm{W}}_{k,k',l}) = 1 \forall k,k' \in {\cal{K}}, l \in {\cal{L}},  \\
	\label{P2.2-f}& \tilde{\bm{W}}_{k,l} \succeq 0, \hat{\bm{W}}_{k,k',l} \succeq 0, \forall k,k' \in {\cal{K}}, l \in {\cal{L}}.
	\vspace{-2mm}  
\end{align} 
\end{subequations} 
It is still challenging to solve problem (\ref{P2.2}) due to non-convex constraints ({\ref{P2.2-b}}), ({\ref{P2.2-c}}), ({\ref{P2.2-e}}), and integer constraint ({\ref{P2-f}}). To deal with this issue, we relax the rank-one constraints to obtain a semidefinite relaxation (SDR) version of problem ({\rm{P2.2}}), denoted by (SDR2.2). In the following, we prove that the optimal solution with rank-one matrices $\{\tilde{\bm{W}}_{k,l}\}$ and $\{\hat{\bm{W}}_{k,k',l}\}$ is also optimal for (P2.2). 
\begin{Pro}\label{RankOne}
	There always exists an optimal solution of problem (\ref{P2.2}) satisfying ${\rm{rank}}(\bar{\bm{W}}_{k,l}) = 1$ and ${\rm{rank}}(\bar{\bm{W}}_{k,k',l}) = 1$ that can be obtained based on the optimal solution of problem (SDR2.2), i.e.,
	\vspace{-2mm}  
	\begin{equation}\label{ConstructionSolution}
		\begin{aligned}
			&\bar{\bm{w}}_{k,l} = \left\{ {\begin{array}{*{20}{c}}
					{({\bm{v}}^{*H}_l {\bm{Q}}_k \tilde{\bm{W}}^*_{k,l} {\bm{Q}}^H_k {\bm{v}}^{*}_l )^{-1/2} \tilde{\bm{W}}^*_{k,l} {\bm{Q}}^H_k {\bm{v}}^{*}_l,}&{ c^*_{k,l} = 1}\\
					{{\bm{0}},}&{{\rm{otherwise}}}
			\end{array}} \right., \\
			&\bar{\bm{W}}_{k,l} = \bar{\bm{w}}_{k,l} \bar{\bm{w}}^H_{k,l}; \quad \bar{c}_{k,l} = {c}^*_{k,l}, \quad \bar{\bm{v}}_l = {\bm{v}}^*_l, \quad \forall k \in {\cal{K}}, l \in {\cal{L}}, 
			\vspace{-2mm}  
		\end{aligned}		
	\end{equation}
	\begin{equation}\label{ConstructionSolution2}
		\bar{\bm{w}}_{k,k',l} = \left\{ {\begin{array}{*{20}{c}}
				{\bar{\bm{w}}_{k,l},}&{ c^*_{k',l} = 1}\\
				{{\bm{0}},}&{{\rm{otherwise}}}
		\end{array}} \right.,  \bar{\bm{W}}_{k,k',l} = \bar{\bm{w}}_{k,k',l} \bar{\bm{w}}^H_{k,k',l},  
	\end{equation}
	where $\{ \{\tilde{\bm{W}}^*_{k,l}\}, \{{\hat{\bm{W}}^*_{k,k',l}}\}, \{{\bm{v}}^*_l\}, \{c^*_{k,l}\} \}$ denotes an optimal solution of (SDR2.2).
\end{Pro}
\begin{proof}
	Similar as in the proof for Proposition \ref{RankOneCosensing}, it can be verified that the solution constructed in (\ref{ConstructionSolution}) always satisfies the target grouping constraints (\ref{P2-e}) and (\ref{P2-f}), the beam pattern constraints (\ref{P2.2-b}), (\ref{P2.2-c}), and the transmit power constraint (\ref{P2.2-d}). Also, ${\rm{rank}}({\tilde{\bm{W}}^*_{k,l}}) = {\rm{rank}}({\hat{\bm{W}}^*_{k,k',l}})$ since ${\tilde{\bm{W}}^*_{k,k',l}} = {\hat{\bm{W}}^*_{k,k',l}}$ when $c^*_{k',l} = 1$. As the detailed proof is similar to that in Appendix A, it is omitted here due to the page limitation.
\end{proof} 

Furthermore, to relax the non-convex constraints regarding the integer variables, we rewrite constraint ({\ref{P2-f}}) as follows:
\vspace{-2mm}  
\begin{equation}\label{BinaryLinearity}
	c_{k,l}( 1 - c_{k,l}) \le 0; \quad 0 \le c_{k,l} \le 1, \forall k \in {\cal{K}}, l \in {\cal{L}}.
	\vspace{-2mm}  
\end{equation}
We can readily show that the $c_{k,l}$ satisfying the above two constraints must be either 1 or 0, which confirms the equivalence of ({\ref{P2-f}}) and (\ref{BinaryLinearity}). Then, the constraints $c_{k,l}( 1 - c_{k,l}) \le 0$, $\forall k \in {\cal{K}}, l \in {\cal{L}}$ in ({\ref{BinaryLinearity}}) are added to the objective function in (SDR2.2) as a penalty term, and problem (SDR2.2) is recast as 
\vspace{-2mm}  
\begin{subequations}\label{P2.3}
\begin{align}
	({\rm{P2.3}}): \quad \label{P2.3-a}&	\mathop {\max }\limits_{\{\tilde{\bm{W}}_{k,l}\}, \{\hat{\bm{W}}_{k,k',l}\}, \{{\bm{v}}_l\}, \{c_{k,l}\}, R} \quad \mathop R - \frac{1}{\eta}\sum\nolimits_{k = 1}^K \sum\nolimits_{l = 1}^L c_{k,l}( 1 - c_{k,l})  \\ 
	\mbox{s.t.}\quad
	& (\rm{\ref{P2-e}}), (\ref{UnitModule}), (\ref{StrackConstraintsA})-(\ref{StrackConstraintsC}), (\rm{{\ref{P2.2-b}}})-(\rm{{\ref{P2.2-d}}}), (\rm{{\ref{P2.2-f}}}),  \nonumber \\
	\label{P2.3-b} & c_{k,l} \in [0,1], \forall k \in {\cal{K}}, l \in {\cal{L}},
	\vspace{-2mm}  
\end{align} 
\end{subequations} 
where $\eta > 0$ is the penalty coefficient used to penalize the violation of equality constraint ({\ref{P2-f}}). It can be readily verified that the solution of problem (\ref{P2.3}) satisfies the inequality constraint in (\ref{BinaryLinearity}) (i.e., the $\{c_{k,l}\}$ are binary values), when $\frac{1}{\eta} \to \infty$. Accordingly, in the inner layer, all optimization variables are partitioned into two blocks, where transmit covariance matrices $\{\tilde{\bm{W}}_{k,l}, \hat{\bm{W}}_{k,k',l}\}$ and target grouping variables $c_{k,l}$ are allocated to one block, while IRS phase shift vectors ${\bm{v}}_l$ are optimized in the other block. Then, the AO method is applied to iteratively optimize the variables in the two blocks.

\subsubsection{Optimization of Transmit Covariance Matrices and Target Grouping}
For given IRS phase shifts, the successive convex approximation (SCA) technique is adopted to obtain the beamforming vector and target grouping. For a given feasible point $c_{k,l}^{(r)}$, we have
\vspace{-2mm}  
\begin{equation}\label{BenaryLangrauage}
	c_{k,l}( 1 - c_{k,l}) \le c_{k,l} -2 c_{k,l} c_{k,l}^{(r)} + \left(c_{k,l}^{(r)}\right)^{2} = {\cal{G}}^{(r)}_{ub}(c_{k,l}), \forall k \in {\cal{K}}, l \in {\cal{L}}.
	\vspace{-2mm}  
\end{equation}
By replacing the corresponding terms $c_{k,l}( 1 - c_{k,l})$ in the objective function of problem (\ref{P2.3}), it can be transformed to 
\vspace{-2mm}  
\begin{subequations}\label{P2.3.1}
	\begin{align}
	(\rm{P2.3.1}): \quad \label{P2.3.1a} & \mathop {\max }\limits_{\{\tilde{\bm{W}}_{k,l}\}, \{\hat{\bm{W}}_{k,k',l}\}, \{c_{k,l}\}, R} \quad \mathop R - \frac{1}{\eta}\sum\nolimits_{k = 1}^K \sum\nolimits_{l = 1}^L {\cal{G}}^{(r)}_{ub}(c_{k,l}) \\ 
	\mbox{s.t.}\quad
	& (\rm{\ref{P2-e}}), (\ref{StrackConstraintsA})-(\ref{StrackConstraintsC}), (\rm{{\ref{P2.2-b}}})-(\rm{{\ref{P2.2-d}}}), (\rm{{\ref{P2.2-f}}}), (\rm{\ref{P2.3-b}}). \nonumber 
	\vspace{-2mm}  
\end{align} 
\end{subequations} 
It is not difficult to show that the objective function and the constraints of problem (\rm{P2.3.1}) are all convex, i.e., the above problem can be solved by existing convex optimization tools such as CVX \cite{Michael2014cvx}.

\subsubsection{Optimization of Phase Shift}
For given transmit covariance matrix and target grouping, the phase shift problem can be formulated as follows:
\vspace{-2mm}  
\begin{subequations}\label{P2.3.2}
	\begin{align}
	(\rm{P2.3.2}): \quad \label{P2.3.2a} & \mathop {\max }\limits_{ \{{\bm{v}}_l\}, R} \quad \mathop R \quad \mbox{s.t.}\quad  (\ref{UnitModule}),  (\rm{{\ref{P2.2-b}}}), (\rm{{\ref{P2.2-c}}}). 
	\vspace{-2mm}  
\end{align} 
\end{subequations} 
In constraint ({\ref{P2.2-b}}), in the $r$th iteration, ${\bm{v}}^H_l {\bm{Q}}_k {\tilde{\bm{W}}_{k,l}} {\bm{Q}}^H_k {\bm{v}}_l$ can be linearized based on the first-order Taylor expansion at a given point ${\bm{v}}_l^{(r)}$, yielding the following inequality
\vspace{-2mm}  
\begin{equation}\label{LowerBoundPhase}
	{\bm{v}}^H_l {\bm{Q}}_k {\tilde{\bm{W}}_{k,l}} {\bm{Q}}^H_k {\bm{v}}_l \ge -({\bm{v}}^H_l)^{(r)} {\bm{Q}}_k {\tilde{\bm{W}}_{k,l}} {\bm{Q}}^H_k {\bm{v}}_l^{(r)} + 2 {\rm{Re}}\left( ({\bm{v}}^H_l)^{(r)} {\bm{Q}}_k {\tilde{\bm{W}}_{k,l}} {\bm{Q}}^H_k {\bm{v}}_l\right).
	\vspace{-2mm}  
\end{equation}
After replacing the corresponding terms in ({\ref{P2.2-b}}) by (\ref{LowerBoundPhase}), optimization problem (\ref{P2.3.2}) in the $r$th iteration can be solved by convex optimization solvers such as CVX \cite{Michael2014cvx}. Finally, to satisfy the unit-modulus constraint in (\ref{P2-d}), let $[\hat{\bm{v}}_l]_{n} = [\tilde{\bm{v}}_l]_{n}/|[\tilde{\bm{v}}_l]_{n}|$, where $\tilde{\bm{v}}_l$ denotes the phase shift vector obtained by solving problem (\ref{P2.3.2}).

\subsubsection{Outer Layer Iteration}
In the outer layer, the value of the penalty coefficient $\eta$ is gradually reduced by updating $\eta^{(r+1)} \to e \eta^{(r)}$, where $e$ ($0 < e < 1$) is a scaling factor. In the inner layer, with the given penalty coefficient, the objective function in (\ref{P2.3-a}) is non-decreasing in each AO iteration. In the outer layer, $\eta$ is initialized to a sufficiently large value and then gradually decreased so that inequality constraint (\ref{BinaryLinearity}) is ultimately satisfied. As such, this penalty-based framework is guaranteed to converge to a locally optimal point of problem (P2.3) based on Appendix B in \cite{Cai2017Joint}. The details of the proposed penalty-based algorithm are shown in {\bf{Algorithm} \ref{PenaltyBasedAlgorithm}}.

The complexity of the proposed penalty-based algorithm can be analyzed as follows. In the inner layer, the main complexity is caused by the computation of $\{\tilde{\bm{W}}_{k,l}\}$, $\{\hat{\bm{W}}_{k,k',l}\}$, and $\{c_{k,l}\}$, i.e., ${\cal{O}}\left((2KLM^2 + K^2LM^2 + 1)^{3.5}\right)$ \cite{zhang2019securing}, where $2KLM^2 + K^2LM^2 + 1$ is the number of variables in problem (\rm{P2.3.1}). Similarly, the complexity of solving problem (\rm{P2.3.2}) is ${\cal{O}}\left((LN^2+1)^{3.5}\right)$ \cite{zhang2019securing}, where $LN^2+1$ denotes the number of variables in problem (\rm{P2.3.2}). Thus, the complexity of each iteration in {\bf{Algorithm} \ref{PenaltyBasedAlgorithm}} is ${\cal{O}}\left( (2KLM^2 + K^2LM^2 + 1)^{3.5} + (LN^2+1)^{3.5}\right)$, and its overall complexity depends on the number of iterations required for reaching convergence in the outer layer.

\begin{algorithm}[t]
	\small
	\caption{Penalty-Based Algorithm}
	\label{PenaltyBasedAlgorithm}
	\begin{algorithmic}[1]
		\STATE {\bf{Initialize}}  $\{\tilde{\bm{W}}_{k,l}\}, \{\hat{\bm{W}}_{k,k',l}\}, \{c_{k,l}\}$, and $\{{\bm{v}}_l\}$, iteration number $r = 1$, convergence accuracies $\epsilon_2$ and $\epsilon_3$, objective value $V^{(r)*}$
		\REPEAT 
		\REPEAT 
		\STATE For given $\{{\bm{v}}^{(r)}_l\}$, obtain $\{\tilde{\bm{W}}^{(r+1)}_{k,l}\}, \{\hat{\bm{W}}^{(r+1)}_{k,k',l}\}, \{c^{(r+1)}_{k,l}\}$ by solving problem (\rm{P2.3.1})
		\STATE For given $\{\tilde{\bm{W}}^{(r+1)}_{k,l}\}, \{\hat{\bm{W}}^{(r+1)}_{k,k',l}\}, \{c^{(r+1)}_{k,l}\}$, obtain $\{{\bm{v}}^{(r+1)}_l\}$ by solving problem (\rm{P2.3.2})
		\STATE Update the objective value $V^{(r+1)*}$ of problem (P2.3) according to the obtained $\{\tilde{\bm{W}}^{(r+1)}_{k,l}\}, \{\hat{\bm{W}}^{(r+1)}_{k,k',l}\}, \{c^{(r+1)}_{k,l}\}$, and $\{c^{(r+1)}_{k,l}\}$
		\STATE $r = r + 1$
		\UNTIL $\left|V^{(r)*} - V^{(r-1)*}\right| \le \epsilon_2$
		\STATE $\eta \to e \eta$
		\UNTIL	the constraint violation in (\ref{BinaryLinearity}) is below a threshold $\epsilon_3$.
		\STATE Obtain $\{\bar {\bm{w}}_{k,l}^*\}$ based on Proposition {\ref{RankOne}}.
		\STATE Recover the optimal beamforming vector ${\bm{w}}_{k,l}$ based on $\{\bar {\bm{w}}^*_{k,l}\}$ and $\{c^*_{k,l}\}$.
	\end{algorithmic}
\end{algorithm}

\begin{table}[h]
	\small
	\centering
	\caption{System Parameters}
	\label{tab1}
	\begin{IEEEeqnarraybox}[\IEEEeqnarraystrutmode\IEEEeqnarraystrutsizeadd{2pt}{1pt}]{v/c/v/c/v/v/c/v/c/v/v/c/v/c/v}
		\IEEEeqnarrayrulerow\\
		&\mbox{Parameter}&&\mbox{Value}&&&\mbox{Parameter}&&\mbox{Value}&&&\mbox{Parameter}&&\mbox{Value}&\\
		\IEEEeqnarraydblrulerow\\
		\IEEEeqnarrayseprow[3pt]\\
		&K&&[1,16]&&&L&&\{1,2,4,8,16\}&&&M&&8&\IEEEeqnarraystrutsize{0pt}{0pt}\\
		\IEEEeqnarrayseprow[3pt]\\
		\IEEEeqnarrayrulerow\\
		\IEEEeqnarrayseprow[3pt]\\  
		&N&& 64 &&&N_x&&8&&&N_y&&8&\IEEEeqnarraystrutsize{0pt}{0pt}\\
		\IEEEeqnarrayseprow[3pt]\\
		\IEEEeqnarrayrulerow\\
		\IEEEeqnarrayseprow[3pt]\\  
		&\varepsilon&&5 \times 10^{-6}&&&P^{\max}&&1 \; \mbox{W}&&&\delta&&0.1 \; \mbox{ms}&\IEEEeqnarraystrutsize{0pt}{0pt}\\
		\IEEEeqnarrayseprow[3pt]\\
		\IEEEeqnarrayrulerow\\
		\IEEEeqnarrayseprow[3pt]\\
		&N_p&&20&&&T_w&&8 \; \mbox{ms} &&&\sigma^2&&-70 \; \mbox{dBm}&\IEEEeqnarraystrutsize{0pt}{0pt}\\
		\IEEEeqnarrayseprow[3pt]\\
		\IEEEeqnarrayrulerow\\
		\IEEEeqnarrayseprow[3pt]\\
		&\epsilon_1&& {10}^{-3} &&&\epsilon_2&& {10}^{-3} &&& \epsilon_3 && {10}^{-6} &\IEEEeqnarraystrutsize{0pt}{0pt}\\
		\IEEEeqnarrayseprow[3pt]\\
		\IEEEeqnarrayrulerow
	\end{IEEEeqnarraybox}  
\end{table} 

\section{Simulation Results}
\label{Simulations}
\par 
In this section, Monte Carlo simulation results are provided for characterizing the performance of the proposed sensing schemes and for gaining insight into the design and implementation of IRS-assisted multi-target sensing systems. In our simulations, we adopt the algorithm developed for TD sensing for $K/L=1$ (c.f. Section \ref{TDSensingMechanism}), the algorithm developed for SS sensing for $L=1$ (c.f. Section \ref{SSMechanism}), and otherwise the penalty-based algorithm for TD-SS sensing (c.f. {\bf{Algorithm} \ref{PenaltyBasedAlgorithm}}). The main system parameters are listed in Table \ref{tab1}. We compare the proposed schemes with three benchmark schemes: \footnote{The per iteration complexity of the algorithms employed for the benchmark schemes is ${\cal{O}}\left( (2KLM^2 + K^2LM^2 + 1)^{3.5} + K(N^2L/K+1)^{3.5}\right)$ for IRSD and ${\cal{O}}\left( (2KLM^2 + K^2LM^2 + 1)^{3.5} + (LN^2+1)^{3.5}\right)$ for IRSB and NIC. The overall complexity also depends on the number of iterations required for convergence. }
\begin{itemize}
	\item {\bf{IRS division (IRSD)}}: The elements at the IRS are divided into $K$ groups, and the IRS elements in each group are optimized for maximization of the beam pattern gain of the corresponding target. 
	\item {\bf{IRS beamforming (IRSB)}}: Four adjacent IRS elements are assigned the same phase shift value to reduce the overhead of IRS phase shift control, where the corresponding solution is obtained based on the proposed penalty-based algorithm.
	\item {\bf{No interference constraints (NIC)}}: Motivated by \cite{song2021joint}, the beam pattern is optimized without interference constraints.
\end{itemize}
\begin{figure*}[t]
	\centering
	\setlength{\abovecaptionskip}{0.cm}
	\subfigure[$K/L = 1$.]
	{	
		\label{figure6a}
		\includegraphics[width=6.7cm]{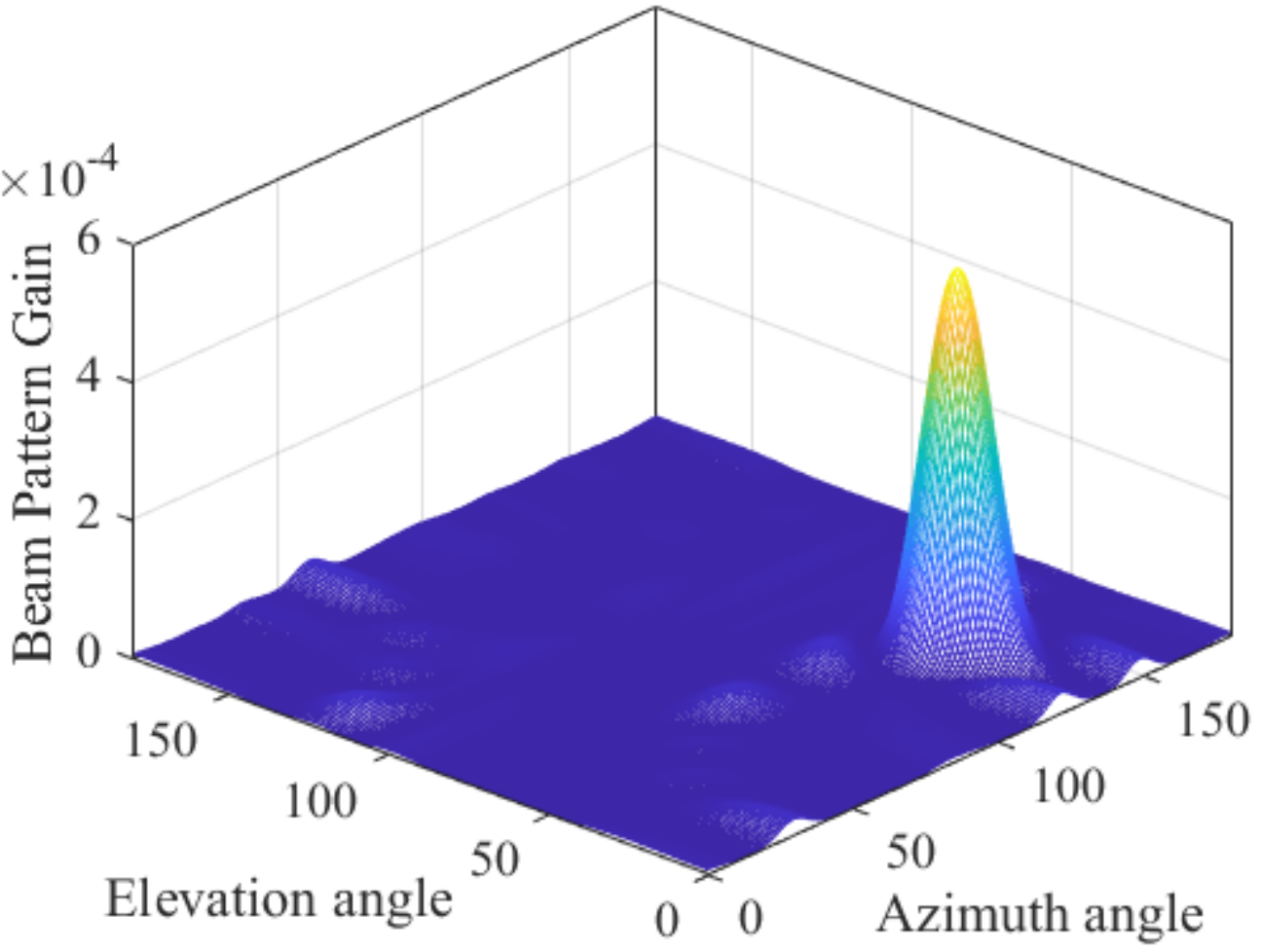}
	}\hspace{15mm}
	\subfigure[$K/L = 2$.]
	{	
		\label{figure6b}
		\includegraphics[width=6.7cm]{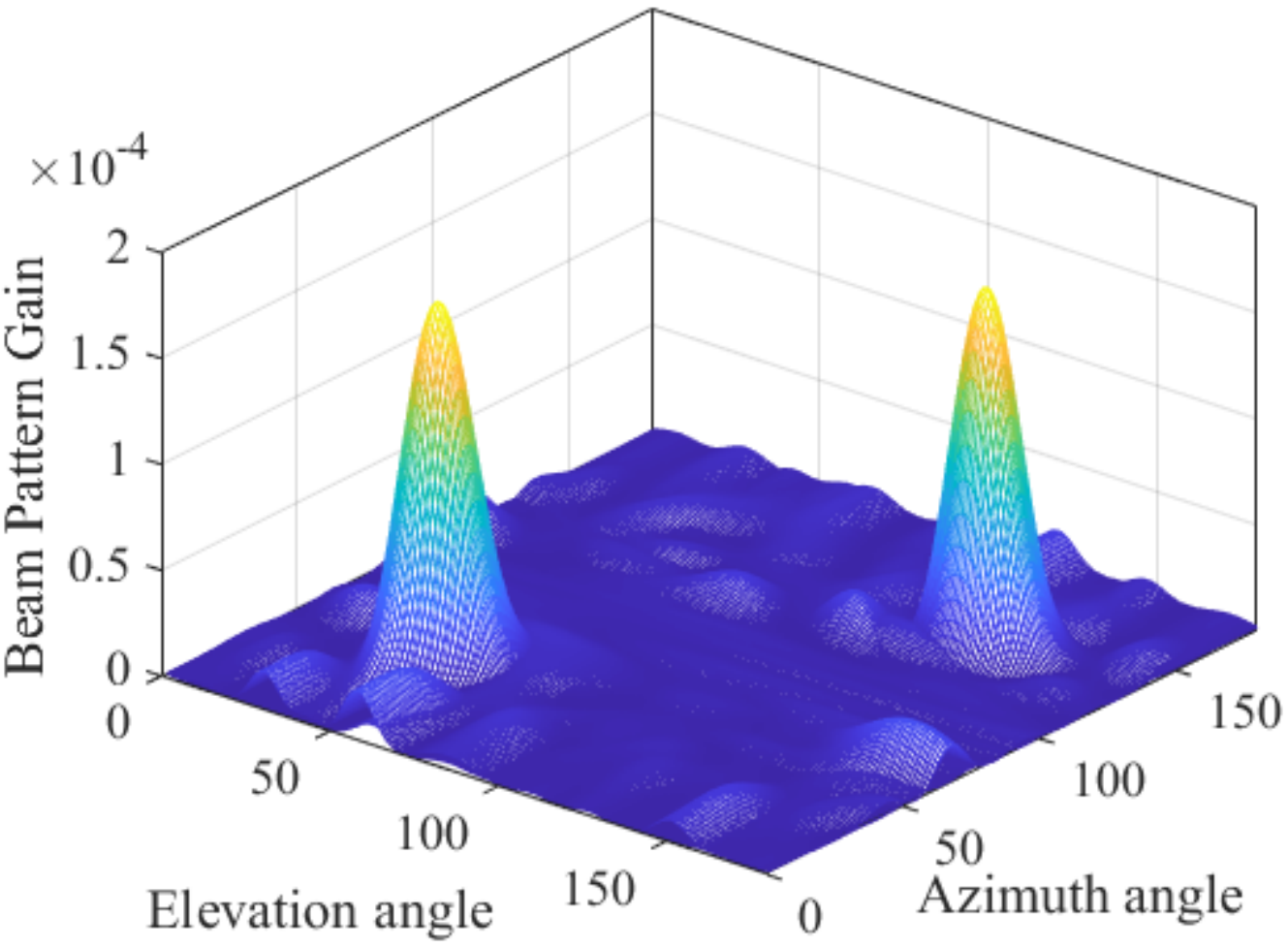}
	}	\\
	\subfigure[$K/L = 4$.]
	{	
		\label{figure6c}
		\includegraphics[width=6.7cm]{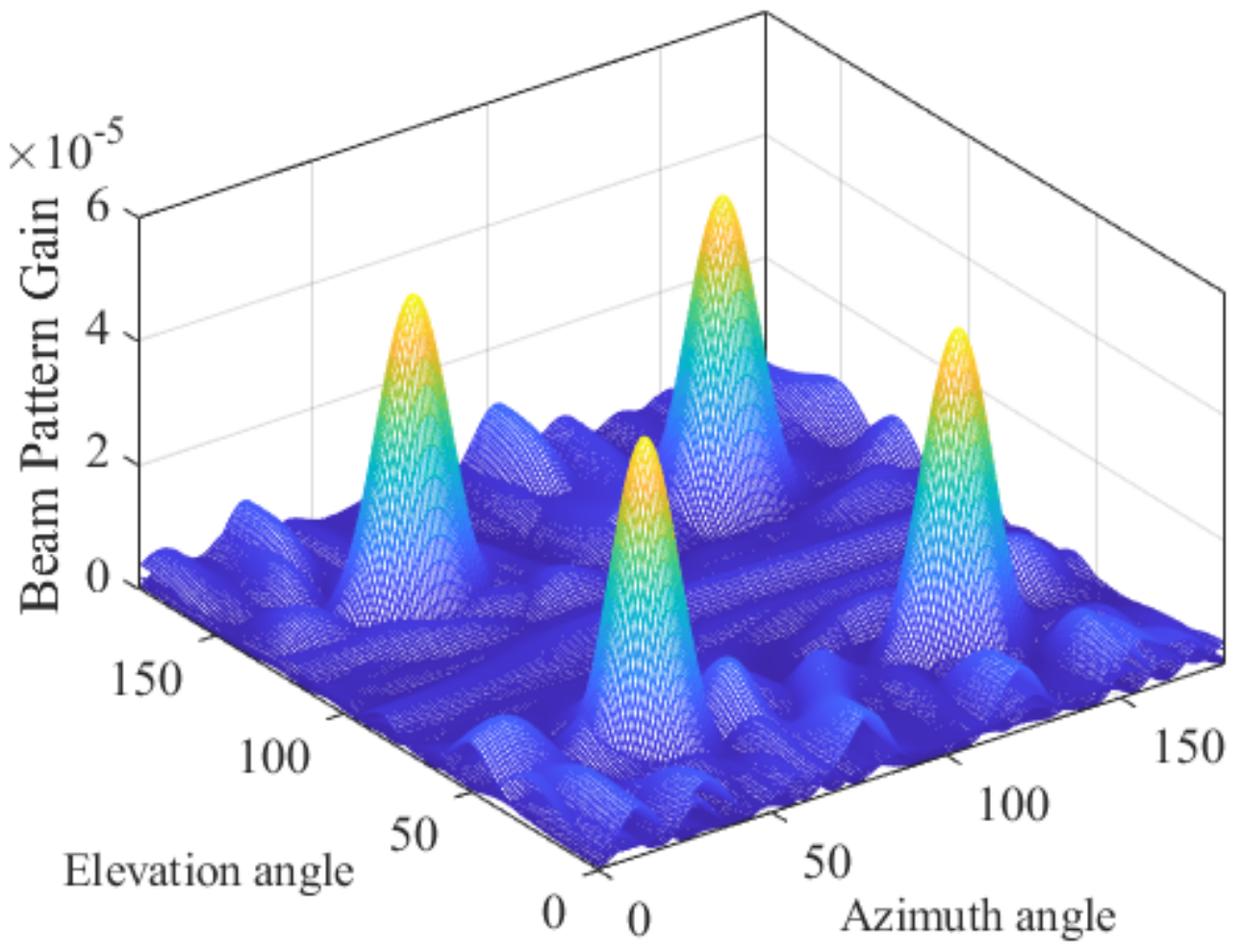}
	}\hspace{15mm}
	\subfigure[$K/L = 8$.]
	{	
		\label{figure6d}
		\includegraphics[width=6.7cm]{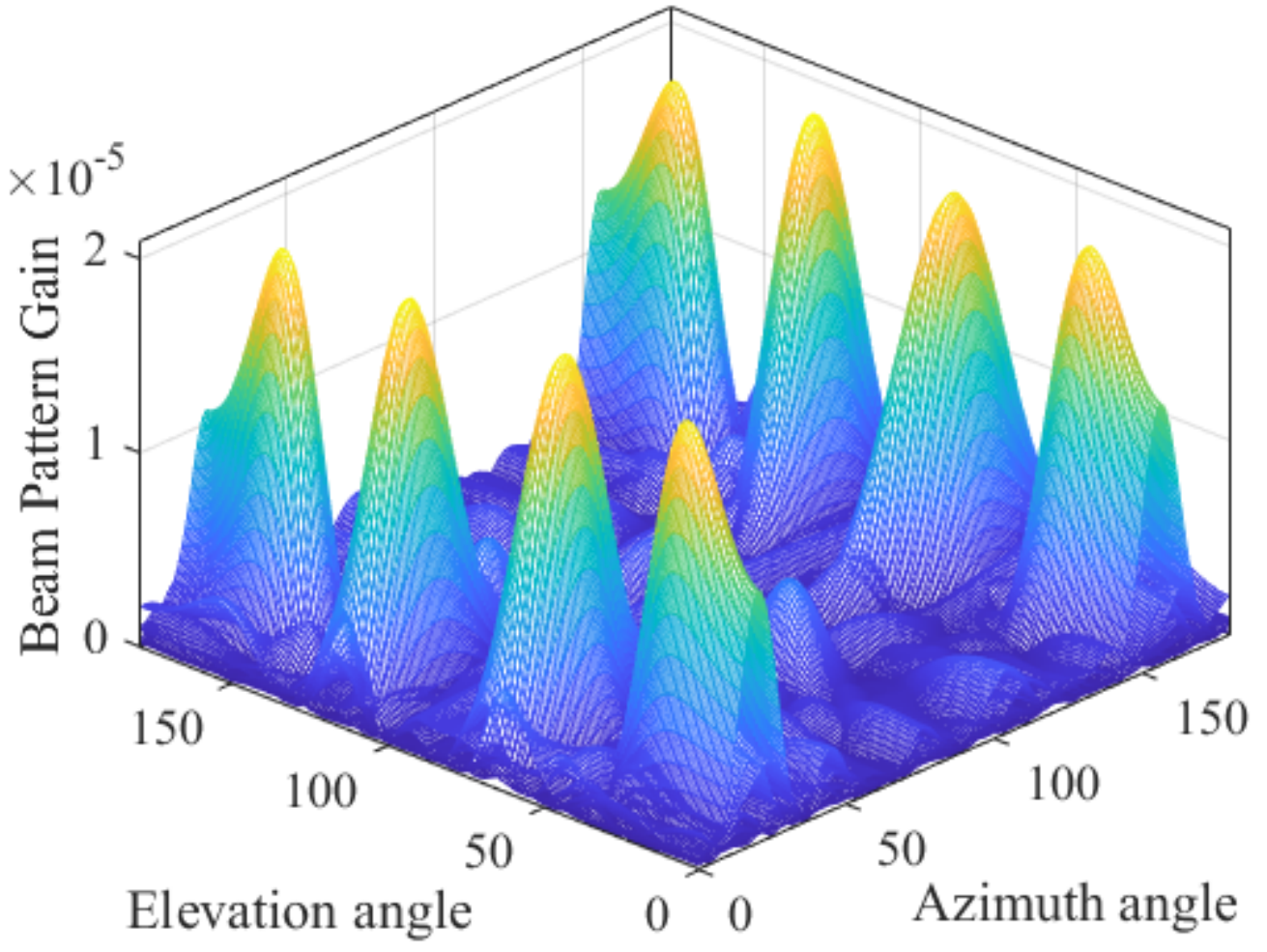}
	}	
	\caption{3D beam pattern gain of the proposed schemes.}
	\label{figure6}
\end{figure*}
\begin{figure*}[t]
	\centering
	\setlength{\abovecaptionskip}{0.cm}
	\subfigure[$K/L = 3$ with $\varepsilon = 5 \times 10^{-5}$.]
	{	
		\label{figure5a}
		\includegraphics[width=6.7cm]{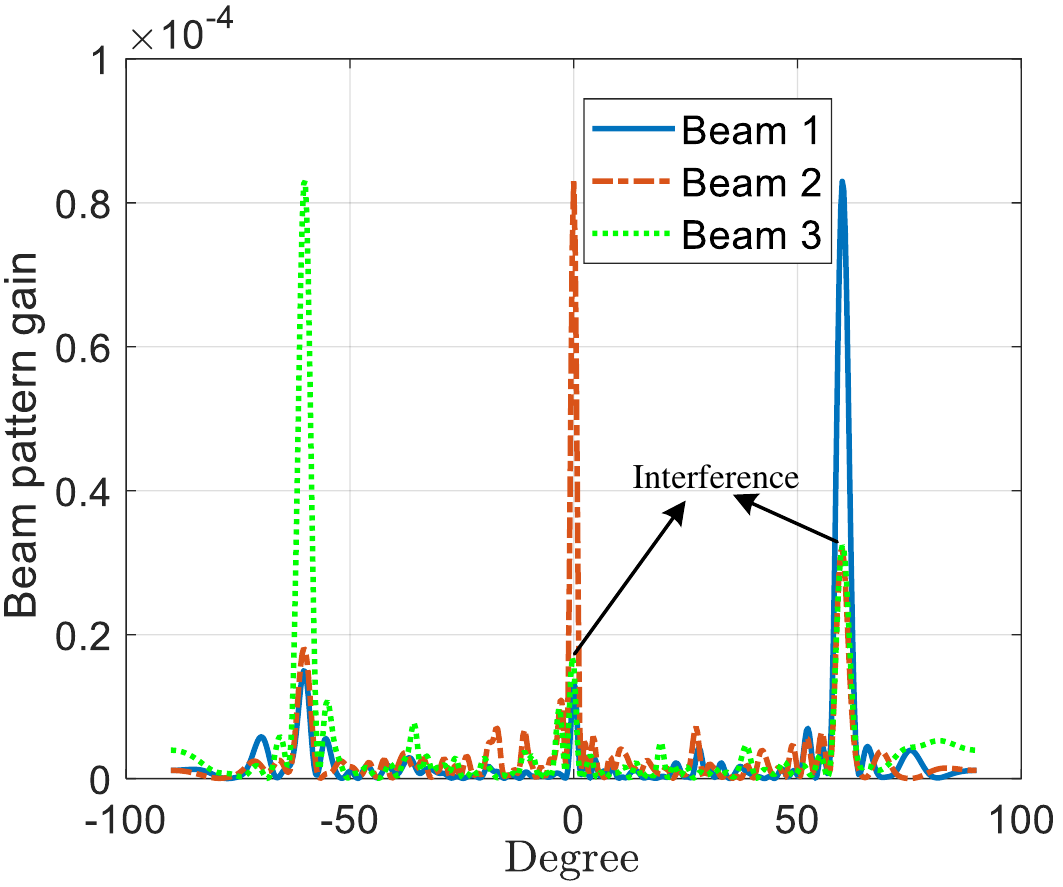}
	}\hspace{15mm}
	\subfigure[$K/L = 3$ with $\varepsilon = 5 \times 10^{-6}$.]
	{	
		\label{figure5b}
		\includegraphics[width=6.7cm]{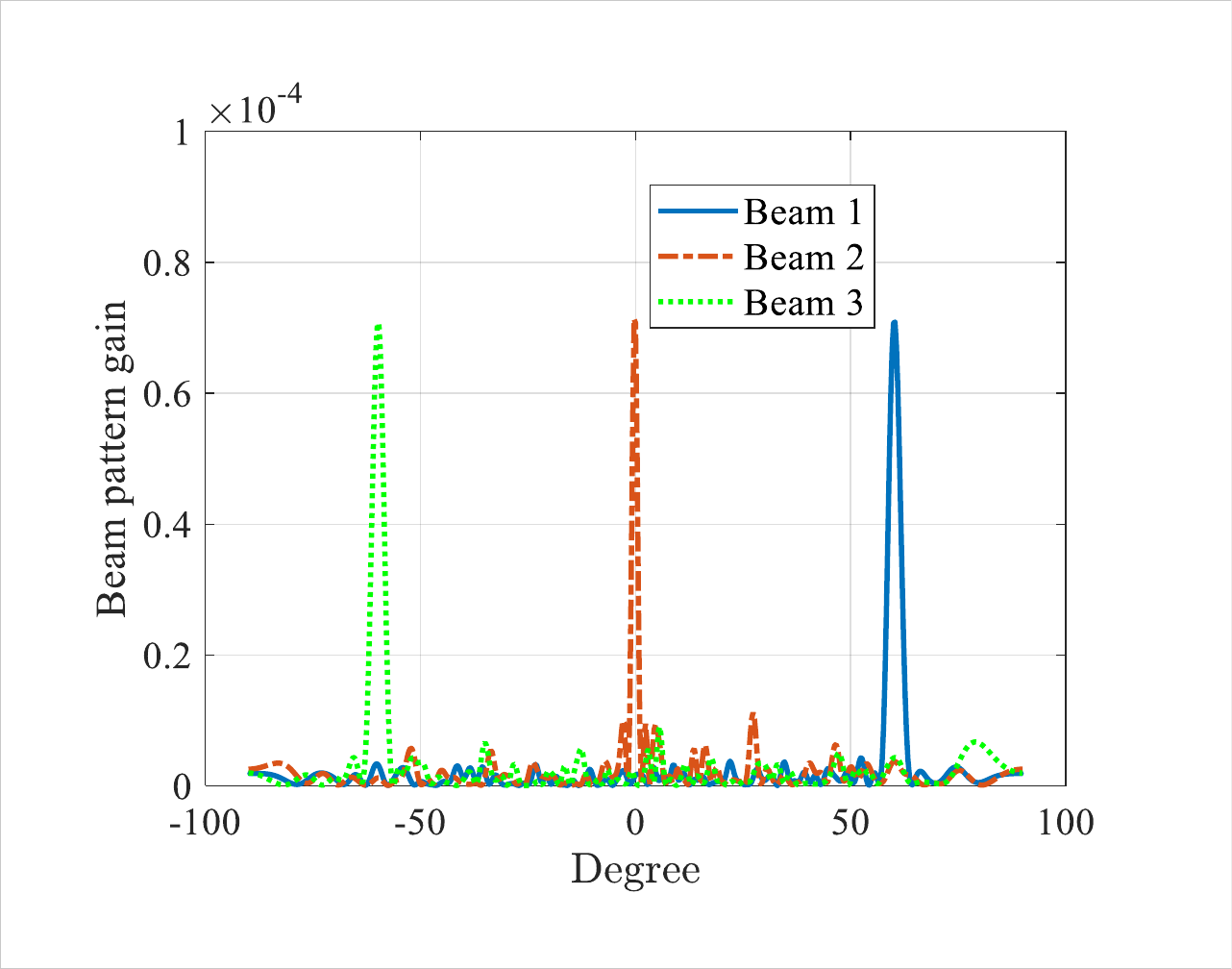}
	} \\
	\subfigure[$K/L = 4$ with $\varepsilon = 5 \times 10^{-5}$.]
	{	
		\label{figure5c}
		\includegraphics[width=6.7cm]{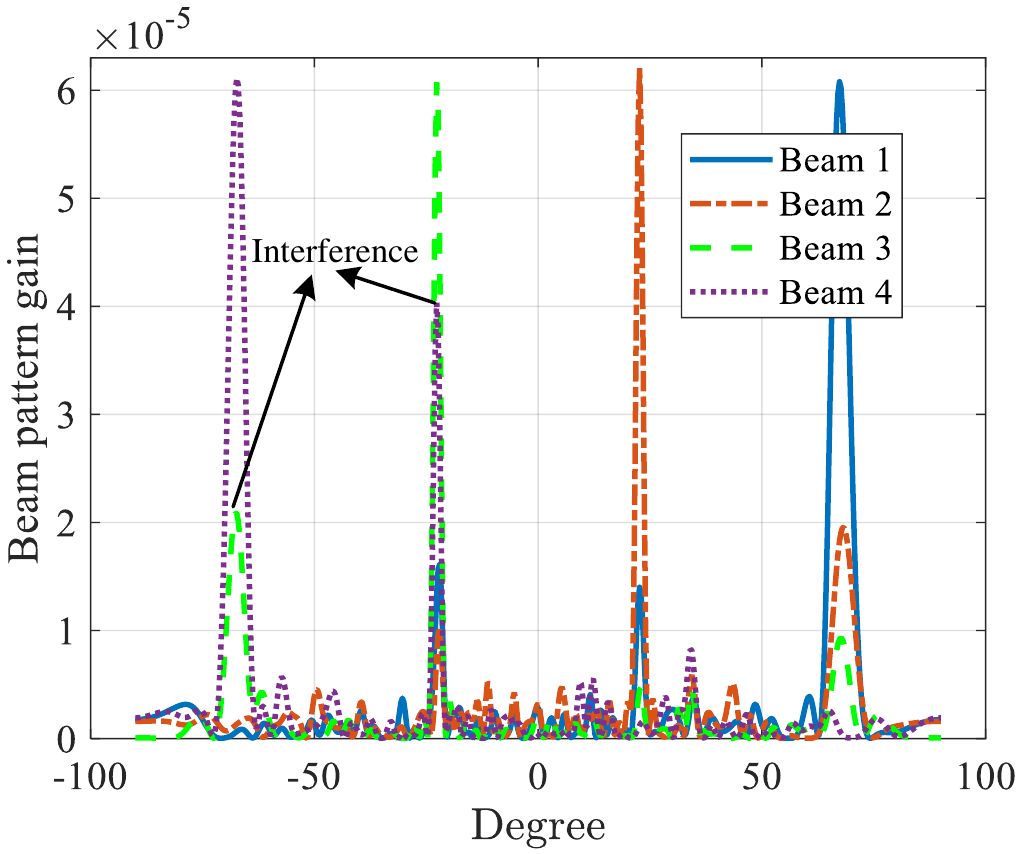}
	}\hspace{15mm}
	\subfigure[$K/L = 4$ with $\varepsilon = 5 \times 10^{-6}$.]
	{	
		\label{figure5d}
		\includegraphics[width=6.7cm]{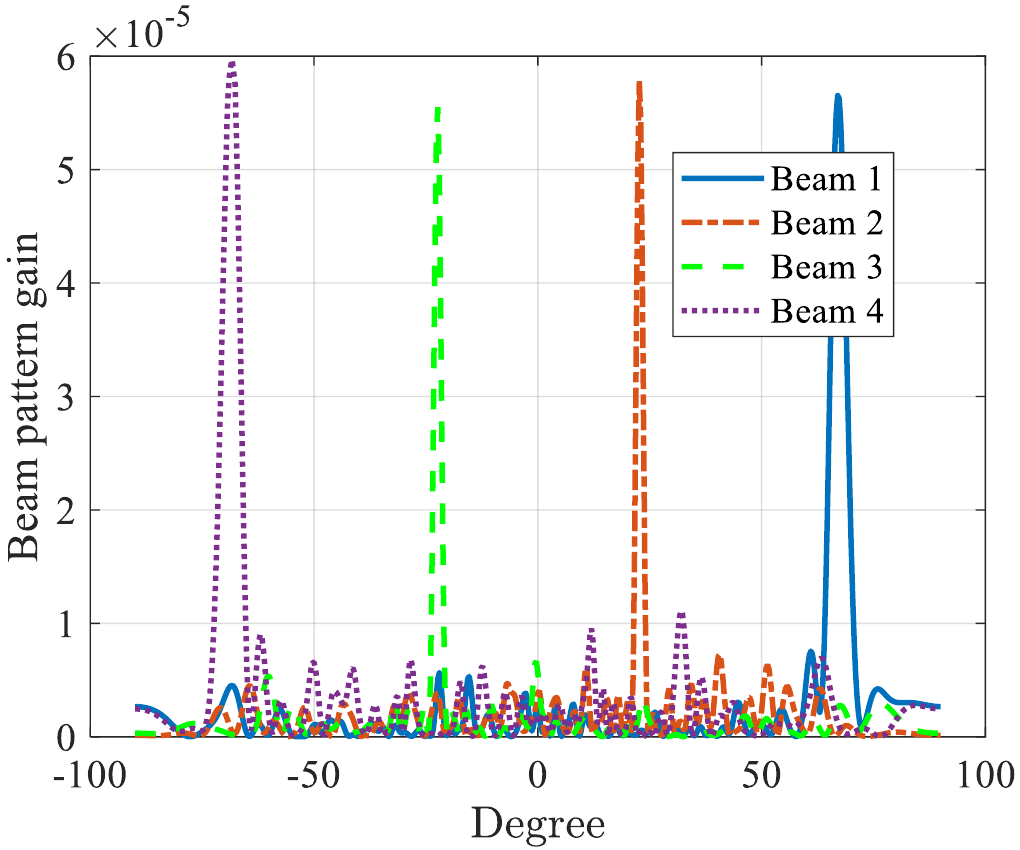}
	}		
	\caption{Beam pattern gain of the proposed schemes for an IRS equipped with a uniform linear array.}
	\label{figure5}
\end{figure*}

\subsection{Beam Pattern Gain}
\subsubsection{Uniform Planar Array}
The 3D beam pattern gains for a uniform planar array are shown in Fig.~\ref{figure6}, where $M = 8$ and $K/L \in \{1,2,4,8\}$. Due to the transmit power and interference constraints, the beam pattern gain decreases for the desired directions as the number of targets in each group increases, while the beam power leaking into side lobes also increases. Besides, when the number of targets in the same group is small, the power of the dedicated beam is more concentrated in the desired direction and relatively symmetrical to the center of the target direction, as can be observed in Figs.~\ref{figure6a} and \ref{figure6b}; when the number of targets is large, the beam pattern gain becomes less symmetrical, as can be seen in Fig.~\ref{figure6d}. 

\subsubsection{Uniform Linear Array} 
In Fig.~\ref{figure5}, the 2-dimensional (2D) beam pattern gain is shown to further illustrate the interference between the different beams for different interference thresholds $\varepsilon$ and different numbers of targets in one group. It is observed from Fig.~\ref{figure5} that for a given number of targets in a group, as the interference threshold $\varepsilon$ decreases from $\varepsilon = 5 \times 10^{-5}$ to $\varepsilon = 5 \times 10^{-6}$, the maximum beam pattern gain is reduced by about $15.6\%$ for $K/L=3$, due to the tighter interference constraints. Furthermore, for more stringent constraints on the mutual interference between different beams, the feasible region of the beamforming vectors and IRS phase shifts becomes more restricted, which may also lead to unequal beam pattern gains for different targets, as can be seen in Fig.~\ref{figure5d}. Moreover, for a given interference constraint, the energy leakage into other directions is larger for four targets in one group compared to three targets in one group. 

\subsection{Trade-off between Beam Pattern Gain and Sensing Frequency}

Fig.~\ref{figure7} investigates the performance trade-off between beam pattern gain, sensing frequency, and transmit power for $M = K = 16$. When $F_s = 6$ Hz ($K = L$), there is only one target in each group, i.e., \textit{TD sensing} is employed, and the benchmark IRSD scheme is equivalent to the proposed scheme. When $F_s = 100$ Hz ($L=1$), all targets are sensed simultaneously, i.e., \textit{SS sensing} is used. It is observed from Fig.~\ref{figure7a} that, as expected from Theorem \ref{SensingPowerVersusFrequency}, the maximum beam pattern gain gradually decreases as the sensing frequency $F_s$ increases. Also, the beam pattern gain achieved by the proposed scheme compared to the benchmark IRSD scheme increases as the sensing frequency increases, since for multi-target sensing scenarios, splitting IRS elements cannot provide a large beam pattern gain and cannot effectively suppress the mutual interference. On the other hand, the beam pattern gain of the benchmark IRSB scheme exceeds that of the benchmark IRSD scheme for large sensing frequencies. Besides, when the sensing frequency increases, the beam pattern gain degrades slightly faster for a higher maximum transmit power $P^{\max}$ than for a lower one. The main reason for this is that a higher transmit power not only improves the effective beam power but also inevitably increases the interference in other directions, thereby resulting in a relatively more limited feasible region for the beamforming vectors and IRS phase shifts. Similarly, it can be seen from Fig.~\ref{figure7b} that the beam pattern gain does not increase linearly with maximum transmit power $P^{\max}$, mainly because larger powers cause stronger side-lobe interference. Moreover, the beam pattern gain increases more slowly when there are more targets in each group as compared to the case when there are fewer targets in each group, due to the tighter interference constraints. 

\begin{figure*}[t]
	\centering
	\setlength{\abovecaptionskip}{0.cm}
	\subfigure[Beam pattern gain versus sensing frequency.]
	{	
		\label{figure7a}
		\includegraphics[width=6.7cm]{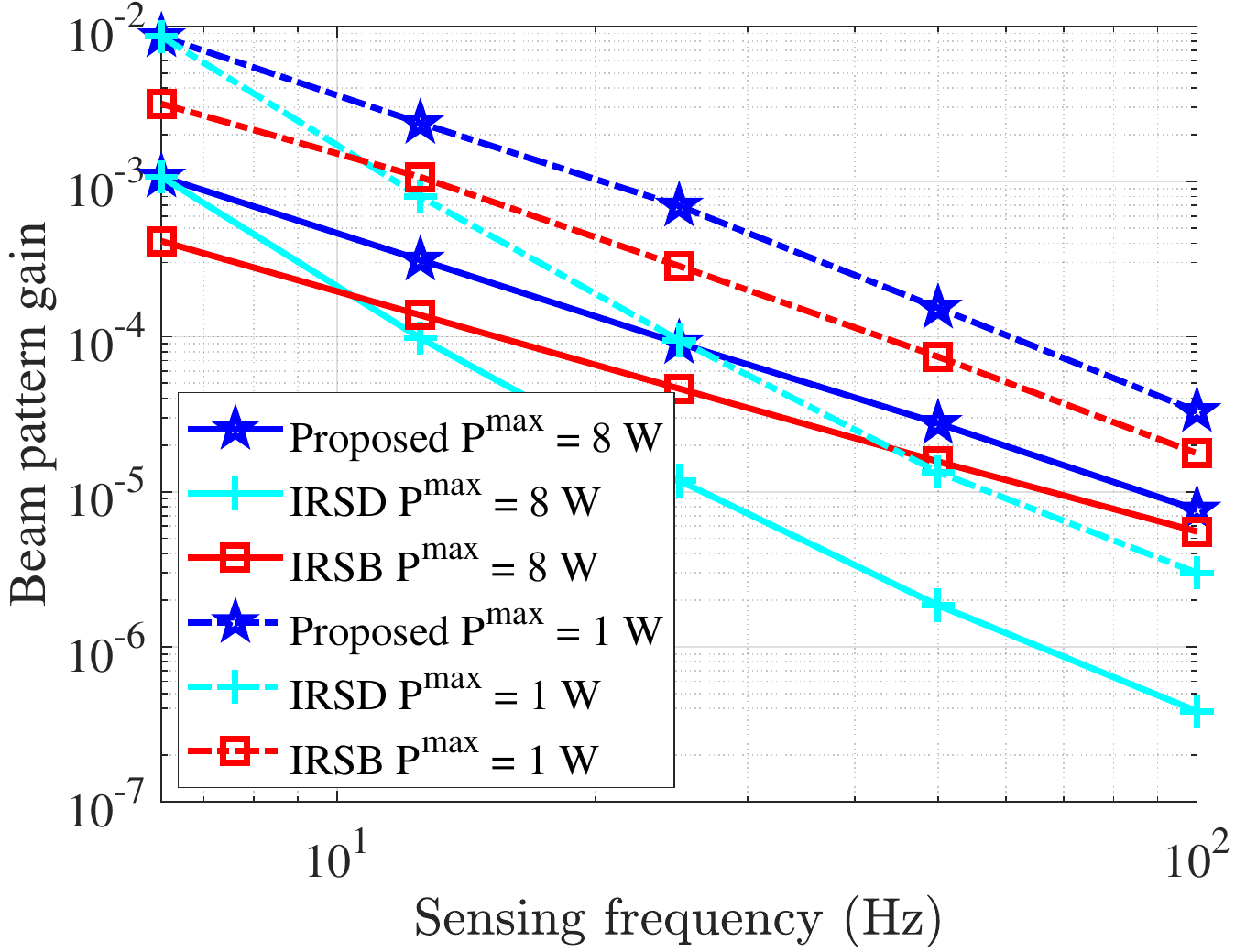}
	}\hspace{15mm}
	\subfigure[Beam pattern gain versus transmit power.]
	{	
		\label{figure7b}
		\includegraphics[width=6.7cm]{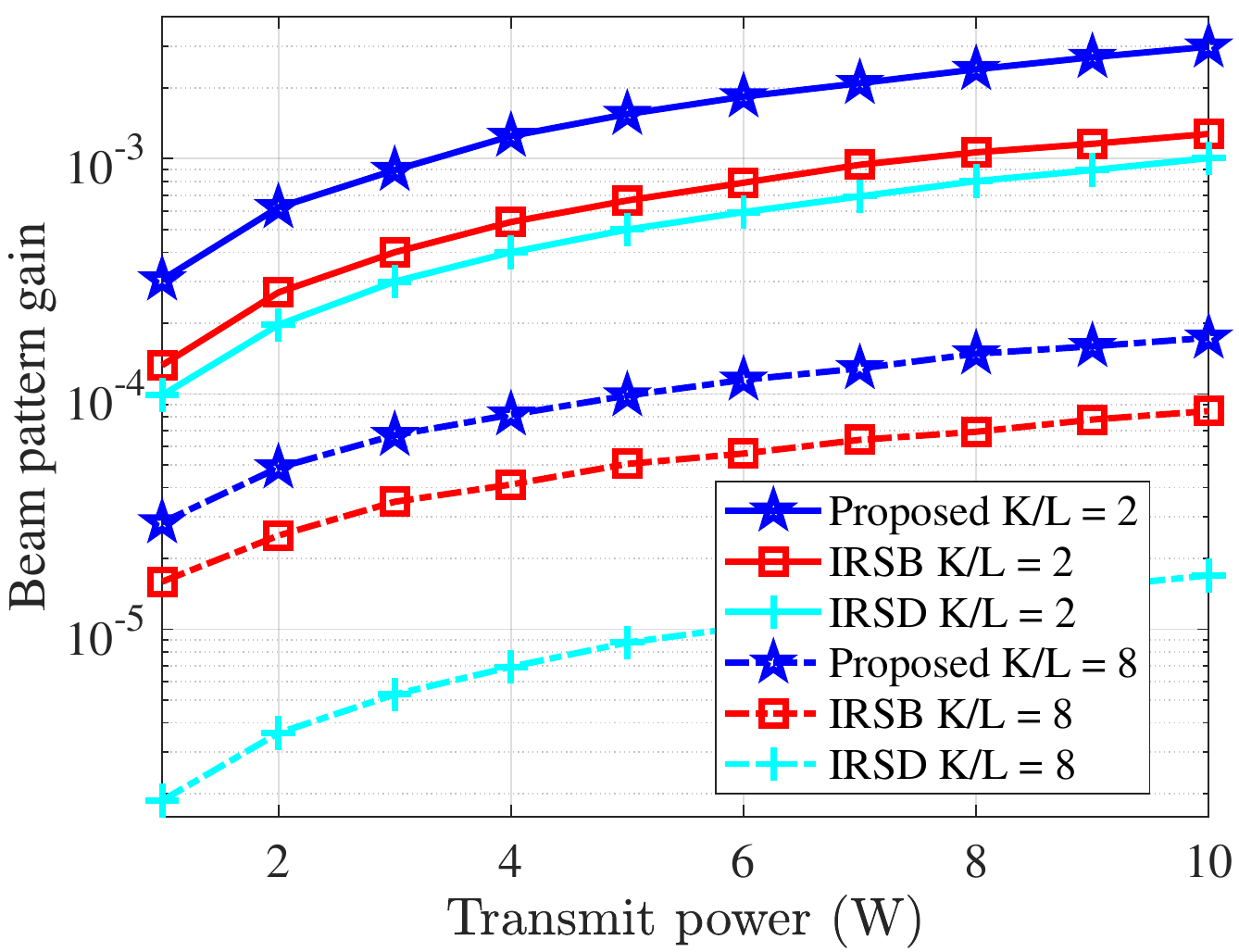}
	}
	\caption{Beam pattern gain versus sensing frequency and transmit power.}
	\label{figure7}
\end{figure*}

\begin{figure*}[t]
	\centering
	\setlength{\abovecaptionskip}{0.cm}
	\subfigure[Beam pattern gain versus number of antennas.]
	{	
		\label{figure8a}
		\includegraphics[width=6.7cm]{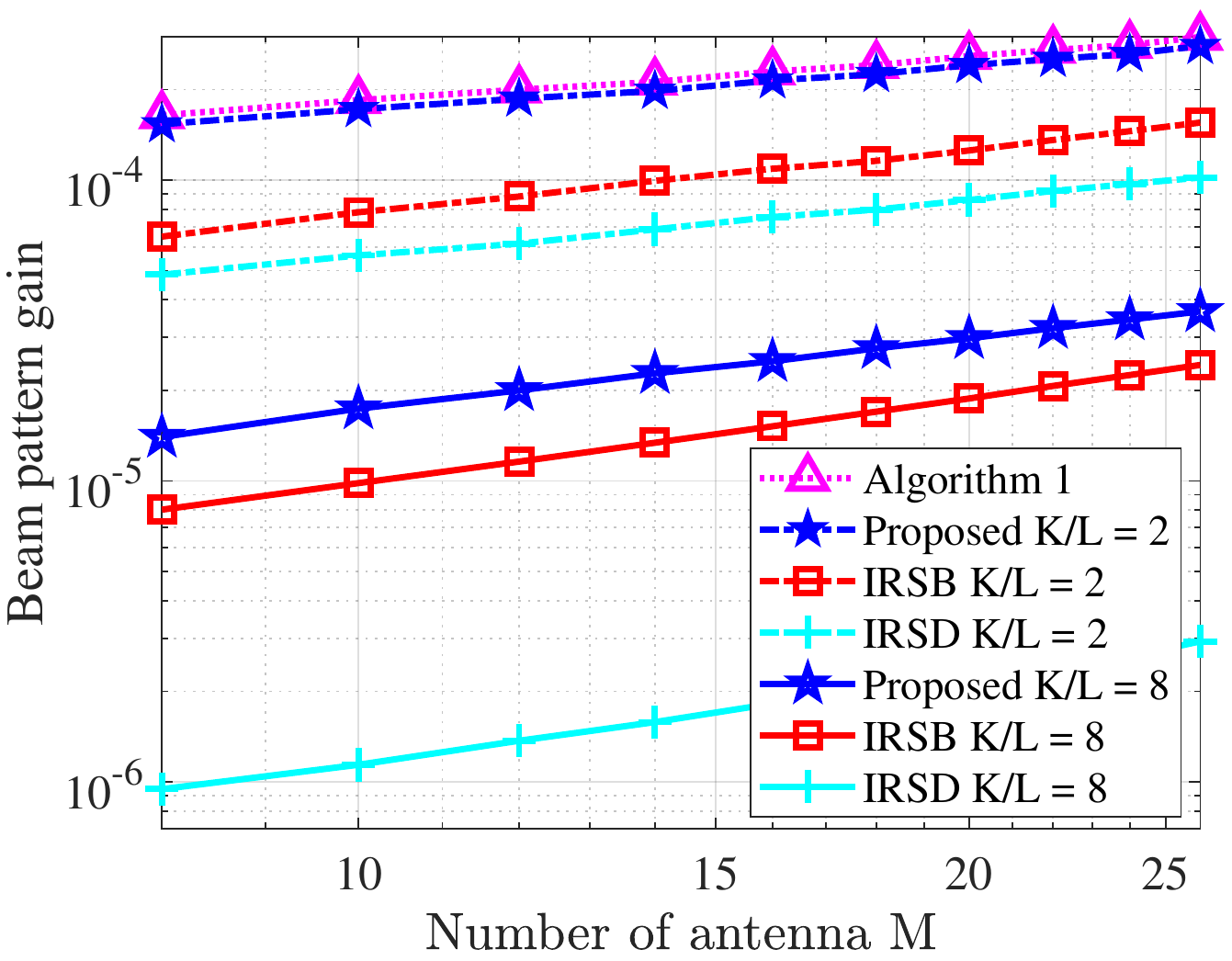}
	}\hspace{15mm}
	\subfigure[Beam pattern gain versus number of IRS elements.]
	{	
		\label{figure8b}
		\includegraphics[width=6.7cm]{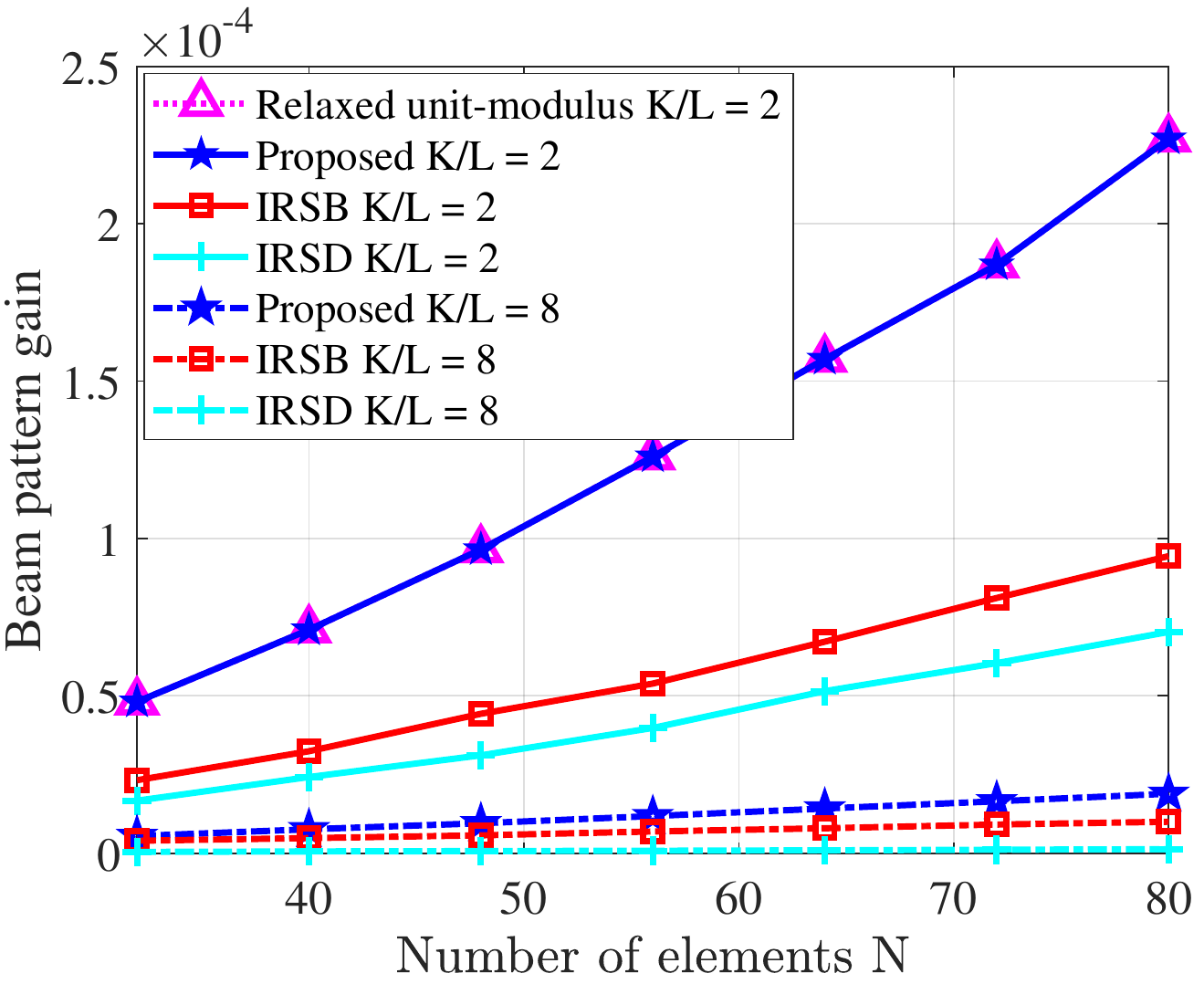}
	}	
	\caption{Beam pattern gain versus numbers of antennas and IRS elements.}
	\label{figure8}
\end{figure*}

\subsection{Impact of Numbers of Antennas and IRS Elements}
In Fig.~\ref{figure8}, we show the beam pattern gains as functions of the number of antennas and the number of IRS elements, respectively. In Fig.~\ref{figure8a}, it can be observed that the beam pattern gain of the three considered schemes increases approximately linearly with the number of antennas. Moreover, the beam pattern gain achieved by the proposed hybrid sensing scheme over the benchmark IRSD scheme for $K/L=8$ is about 4.6 times higher than that for $K/L = 2$, since the beam pattern gain is reduced greatly due to the excessive division of IRS elements in IRSD. Besides, it can be found that as the number of antennas increases, the beam pattern gain of the three considered schemes increases slightly faster if more targets are included in each group compared to the case with fewer targets per group, because it is more difficult to guarantee the interference constraints with a smaller number of antennas. For two targets per group, the beam pattern gain achieved by Algorithm 1, which is based on the derived closed-form beamforming vector in (27), is slightly higher than that of the penalty-based algorithm. Also, as shown in Fig.~\ref{figure8b}, as the number of IRS elements increases, the beam pattern gain of the considered schemes for $K/L = 2$ targets in a group increases faster than that for $K/L = 8$, since the beam pattern gain can be concentrated on fewer designated targets.  
\begin{figure}[!t]
	\begin{minipage}[t]{0.5\linewidth}
		\centering
		\includegraphics[width=6.7cm]{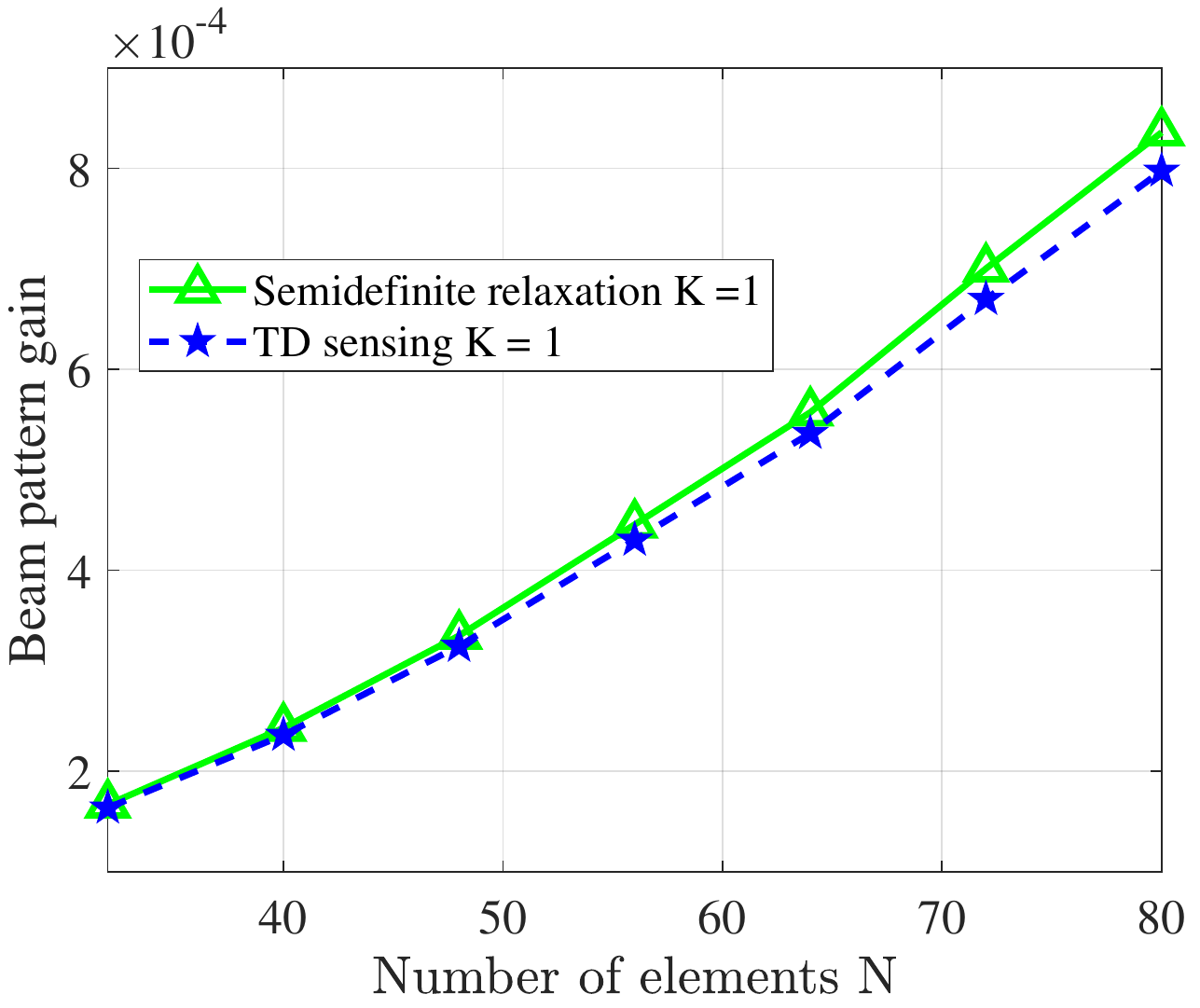}
		\caption{Comparison between the solution for TD sensing \\
			and an upper bound based on semidefinite relaxation.}
		\label{figure9a}
	\end{minipage}%
	\begin{minipage}[t]{0.5\linewidth}
		\centering
		\includegraphics[width=6.7cm]{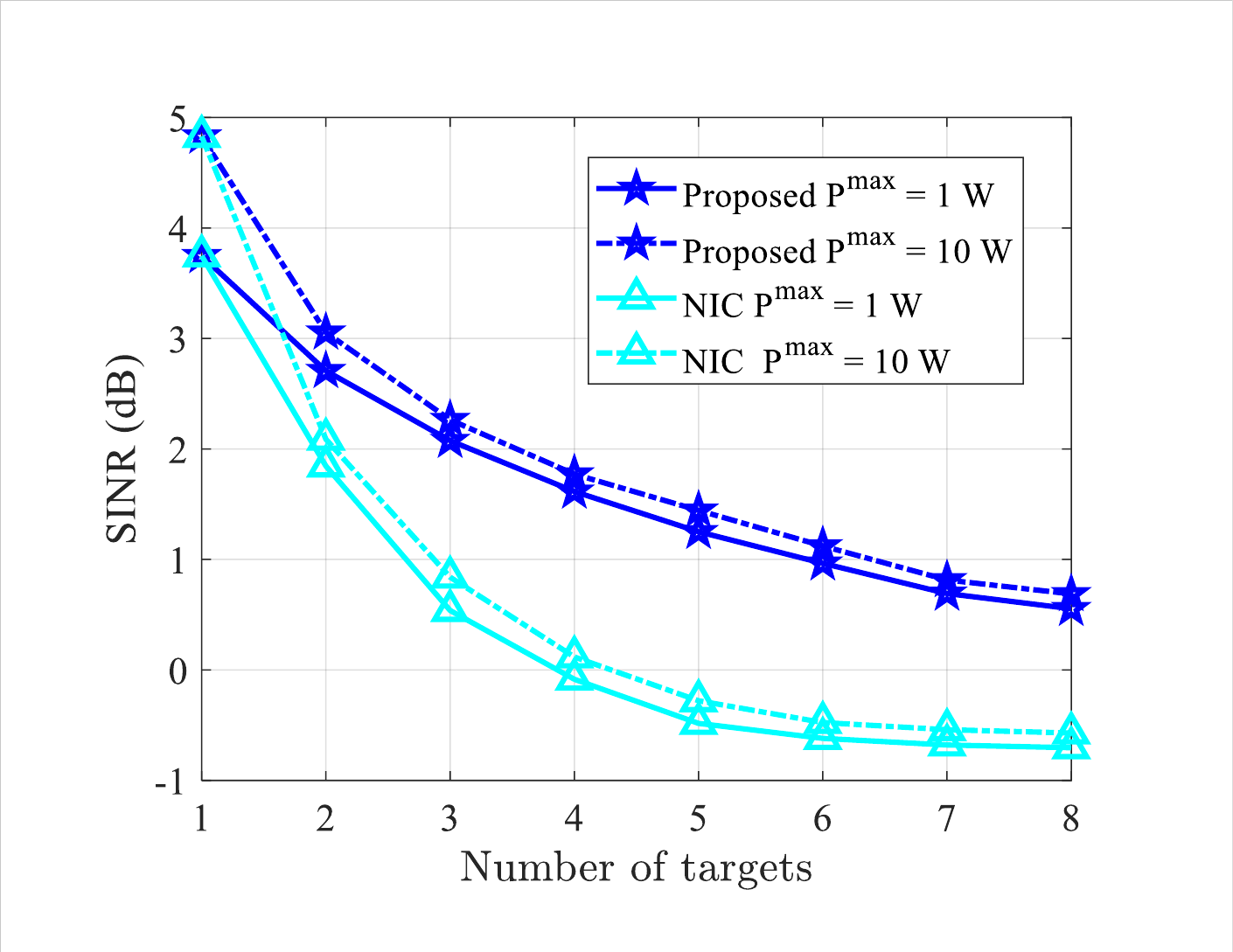}
		\caption{SINR comparison versus number of targets.}
		\label{figure9b}
	\end{minipage}
\end{figure}

\subsection{Evaluation of Relaxations and Approximations}
In this section, we evaluate the tightness of some of the adopted relaxations and approximations. In Fig.~\ref{figure8b}, we observe that for hybrid TD-SS sensing, the beam pattern gain achieved by Algorithm 2 almost coincides with the beam pattern gain obtained for the relaxed unit-modulus constraint in (\ref{UnitModule}) (i.e., the unit-modulus constraint is not enforced in the algorithm), i.e., the performance loss incurred by this relaxation is negligible. Furthermore, in Fig.~\ref{figure9a}, we compare the solution obtained for the TD sensing scheme and an upper bound obtained by solving problem (P1-TD.2) in the absence of the unit-modulus constraints. As can be observed in Fig.~\ref{figure9a}, the performance loss is negligible when the number of IRS elements is small, which means that the approximation in (\ref{RelaxingRankoneConstraintsB}) is tight in this case. When the number of IRS elements increases to 80, the performance loss compared to the upper bound is only 4.6\%.


\subsection{SINR Comparison}
To verify the effectiveness of limiting the power leakage for target detection, we compare the signal-to-interference-noise-ratio (SINR) of the reflected signals of the proposed SS sensing scheme to that of the benchmark NIC scheme. The SINR of target $k$ is ${\rm{SINR}}_k = \frac{{| 	 {\beta}_k|^2|{\bm{f}}_{k}{{\bm{P}}_k}{{\bm{w}}_k}|^2}}{{{{\bm{f}}_{k}^H\left( {{{\bm{R}}_k} + \sigma^2{{\bm{I}}_M}} \right){\bm{f}}_{k}}}}$. where ${\bm{P}}_k \!=\! {\bm{G}}^H {\bm{\Theta}}^H {\bm{a}}_k  {\bm{a}}^H_k {\bm{\Theta}} {\bm{G}}$ and ${\bm{R}}_k = \sum\nolimits_{k' \in {\cal{K}} \backslash k} | \beta_{k'}|^2 {\bm{P}}_{k'} {\bm{w}}_k {\bm{w}}^H_k {\bm{P}}^H_{k'}$. Fig.~\ref{figure9b} shows that the SINR decreases as the number of targets increases due to the decreased beam pattern gain and the increased interference. Besides, the SINR gain achieved by the proposed sensing scheme over the benchmark NIC scheme increases from about 1 dB to 1.67 dB when the number of targets increases from $K=2$ to $K=5$.

\section{Conclusions and Future Works}
\label{Conclusion}
In this paper, we investigated IRS-assisted multi-target sensing, and proposed corresponding \textit{TD}, \textit{SS}, \textit{hybrid TD-SS sensing} schemes. Accordingly, the transmit beamformer at the BS, the IRS phase shifts, and the target grouping were jointly optimized to maximize the beam pattern gain while satisfying interference constraints. For the two-target case, the optimal beamforming vector was derived in closed form to reduce the beamforming design complexity. We proved that rank-one beamforming matrices can always be constructed, and an inner approximation algorithm was presented to obtain a locally optimal solution for the proposed SS scheme. Furthermore, by dividing targets into several disjoint groups, a flexible trade-off between beam pattern gain and sensing frequency was achieved. The numerical results validated the efficiency of the proposed designs over three benchmark schemes. Our results also reveal that as the interference constraints become stricter and/or the number of targets in each group increases, the maximum beam pattern gain decreases while the power leakage in unintended directions increases. 

IRS-assisted integrated sensing and communication in environments with multiple users and multiple targets is an interesting topic for future research. In addition, efficient target grouping for multi-IRS collaborative sensing deserves further investigation.

\section*{Appendix A:  \textsc{Proof of Proposition \ref{RankOneCosensing}}} 

The obtained optimal solution $\{ \{{\bm{W}}^*_{k}\}, {\bm{V}}^*\}$ of (P1-SS-SDR) generally satisfies ${\rm{rank}}({\bm{W}}^*_{k})$  $\ge 1$, $\forall k$. As ${\rm{rank}}({\bm{V}}^*) = 1$, the optimal reflection-coefficient vector ${\bm{v}}^*$ can be recovered by performing eigenvalue decomposition over rank-one matrix ${\bm{V}}^*$. Accordingly, it can be readily proved that the solution $\{ \{\bar{\bm{W}}_{k}\}, \bar{\bm{V}} \}$ satisfies the transmit power constraint (\ref{P1-Co.1-d}). In the following, we prove that constraints ({\ref{P1-Co.1-b}}) and ({\ref{P1-Co.1-c}}) are also satisfied. First, the interference caused by the beam for target $k$ in other directions satisfies 
\vspace{-2mm}
\begin{equation}\label{ConditionCosensingVarepsilon}
	\begin{aligned}
		\sum\nolimits_{k'\in {\cal{K}} \backslash k }{\rm{Tr}}\left({\bar{\bm{W}}_{k}} {\bm{Q}}^H_{k'} \bar{\bm{V}} {\bm{Q}}_{k'}   \right) &{=}  \sum\nolimits_{k'\in {\cal{K}} \backslash k } {\rm{Tr}}\left({\bm{Q}}^H_{k'} {\bm{v}}^{*} {\bm{v}}^{*H} {\bm{Q}}_{k'}{\bm{W}}^*_{k} \frac{  {\bm{Q}}^H_k {\bm{v}}^{*}  {\bm{v}}^{H*} {\bm{Q}}_k {\bm{W}}^{*H}_{k} }{{\bm{v}}^{*H} {\bm{Q}}_k {\bm{W}}^*_{k} {\bm{Q}}^H_k {\bm{v}}^{*} } \right) \\
		&\overset{(b)}{\le}  \sum\nolimits_{k'\in {\cal{K}} \backslash k } {\rm{Tr}}\left({\bm{Q}}^H_{k'} {\bm{v}}^{*} {\bm{v}}^{*H} {\bm{Q}}_{k'} {\bm{W}}^*_{k} \right) {\rm{Tr}}\left(\frac{  {\bm{Q}}^H_k {\bm{v}}^{*}  {\bm{v}}^{H*} {\bm{Q}}_k {\bm{W}}^{*H}_{k} }{{\bm{v}}^{*H} {\bm{Q}}_k {\bm{W}}^*_{k} {\bm{Q}}^H_k {\bm{v}}^{*} } \right) \\
		& = \sum\nolimits_{k'\in {\cal{K}} \backslash k } {\rm{Tr}}\left({\bm{Q}}^H_{k'} {\bm{v}}^{*} {\bm{v}}^{*H} {\bm{Q}}_{k'} {\bm{W}}^*_{k} \right)  \le \varepsilon.
		\vspace{-2mm}	
	\end{aligned}
\end{equation} 
In (\ref{ConditionCosensingVarepsilon}), ($b$) holds for the following reasons. For notational convenience, define ${\bm{t}}_{k'} = {\bm{Q}}^H_{k'} {\bm{v}}^{*}$ and ${\bm{T}}_{k'} = {\bm{t}}_{k'}{\bm{t}}^H_{k'}$. First, since ${{\bm{T}}_{k'}}^{1/2} {\bm{W}}^*_{k} {{\bm{T}}_{k'}}^{1/2}$ is symmetric and positive semidefinite, and ${\bm{T}}_{k'} {\bm{W}}^*_{k}$ and ${{\bm{T}}_{k'}}^{1/2} {\bm{W}}^*_{k} {{\bm{T}}_{k'}}^{1/2}$ have the same non-negative eigenvalues, the eigenvalues of ${\bm{Q}}^H_k {\bm{v}}^{*}  {\bm{v}}^{H*} {\bm{Q}}_k {\bm{W}}^{*H}_{k}$ are non-negative. Second, letting ${\bm{A}} = {\bm{T}}_{k'} {\bm{W}}^*_{k}$ and ${\bm{B}} = {\bm{T}}_{k} {\bm{W}}^*_{k}$, we have ${\rm{Tr}}({\bm{A}}{\bm{B}}) \le {\rm{Tr}}({\bm{A}}){\rm{Tr}}({\bm{B}})$, due to the von Neumann's trace inequality ${\rm{Tr}}({\bm{A}}{\bm{B}}) \le \sum\nolimits_{n = 1}^N \sigma_n({\bm{A}})\sigma_n({\bm{B}})$ \cite{mirsky1975trace}. Similarly, it can be easily shown that 
\begin{equation}
		{\rm{Tr}}\left({\bar{\bm{W}}_{k}} {\bm{Q}}^H_k \bar{\bm{V}} {\bm{Q}}_k   \right) =   {\rm{Tr}}\left({\bm{v}}^{*H} {\bm{Q}}_k \bar{\bm{w}}_{k} \bar{\bm{w}}^H_{k} {\bm{Q}}^H_k {\bm{v}}^{*} \right)  \overset{(d)}{=}  {\rm{Tr}}\left({\bm{v}}^{*H} {\bm{Q}}_k {\bm{W}}^*_{k} {\bm{Q}}^H_k {\bm{v}}^{*}\right) \ge R,
\end{equation}
where ($d$) holds based on the constructed $\bar{\bm{w}}_{k}$ in (\ref{ConstructionRank1Solution}) and $\bar{\bm{W}}_{k} \succeq 0$ can be uniquely determined based on the constructed solution of problem (P1-SS), as provided in (\ref{ConstructionSolution}) and (\ref{ConstructionSolution2}).  
\par
Therefore, the constructed rank-one solution satisfies all constraints of problem (P1-SS). This completes the proof.

\section*{Appendix B: \textsc{Proof of Proposition \ref{TargetSensingStructure}}} 
Let the beamforming vector ${\bm{w}}_k$ have the following structure:
\begin{equation}\label{ConstructionBeamforming}
	{\bm{w}}_k = \rho_1 {\bm{h}}_k + \rho_2 \chi {\bm{H}}_{k'}^{\bot} {\bm{h}}_k + \sum\nolimits_{j \in {\cal{K}} \backslash \{k,k'\}} {\bm{g}}^{\bot}_j,
\end{equation}
where ${\bm{h}}_k = {\bm{G}}^H {\bm{\Theta}}^H_l {\bm{a}}_k $, ${\bm{g}}^{\bot}_j$ is the basis of the null space of $\{{\bm{h}}_k, {\bm{h}}_{k'}\}$, i.e., ${\bm{h}}^H_k {\bm{g}}^{\bot}_j = 0$ and ${\bm{h}}^H_{k'} {\bm{g}}^{\bot}_j = 0$, $\chi$ is a complex number, and $|\chi| = 1$. In (\ref{ConstructionBeamforming}), $\rho_1$ and $\rho_2$ are real numbers. According to the definition in (\ref{NullSpace}), ${\bm{H}}_{k'}^{\bot} {\bm{h}}_k$ lies in the space spanned by vectors ${\bm{h}}_k$ and ${\bm{h}}_{k'}$, but orthogonal to vector ${\bm{h}}_{k'}$. It is easy to verify that ${\bm{h}}_k$ and ${\bm{h}}_{k'}$ help improve the signal-to-interference-noise-ratio (SINR) of the beam pattern gain towards target $k$ while ${\bm{g}}^{\bot}_j$ does not. Hence, the component belonging to the basis of the null space of $\{{\bm{h}}_k, {\bm{h}}_{k'}\}$ can be removed without reducing the performance of the beam pattern gain. Hence, the optimal beamforming vector can be expressed as follows ${\bm{w}}_k^* = \rho_1 {\bm{h}}_k + \rho_2 \chi \frac{{\bm{H}}_{k'}^{\bot} {\bm{h}}_k}{\|{\bm{H}}_{k'}^{\bot} {\bm{h}}_k\|}$. 
First, if $ \frac{{{{p_{k}}|{\bm{h}}^H_{k'}{\bm{h}}_{k}|}}}{\|{\bm{h}}^H_{k}\|} \le  \varepsilon$, the optimal beamforming vector is MRT, i.e., ${\bm{w}}_k^* = \sqrt{p_k}\frac{{\bm{h}}_{k}}{\|{\bm{h}}_{k}\|}$; otherwise, according to constraint (\ref{P1-SS-b}), $\rho_1 = \frac{\sqrt{\varepsilon}}{|{\bm{h}}^H_{k'} {\bm{h}}_k|}$. Then, problem ({P1-SS}) can be reduced to
\begin{equation}\label{P1-TDProof}
	({\rm{P3}}) \quad  \mathop {\max }\limits_{{\bm{w}}_{k}} \quad  {{{\left| {{\bm{h}}_k^H\left( {\rho_1 {{\bm{h}}_k} + \rho_2 \chi {\bm{h}}_{k,k'}^{\bot} } \right)} \right|}^2}}, \quad
	\mbox{s.t.}\quad {{{\left\| {\rho_1 {\bm{h}}_k + \rho_2 \chi {\bm{h}}_{k,k'}^{\bot} } \right\|}^2} \le {p_k}}, 
\end{equation} 
where ${\bm{h}}_{k,k'}^{\bot} = {\bm{H}}_{k'}^{\bot} {\bm{h}}_k$. The objective function can be rewritten as follows:
\begin{equation}\label{ObjectiveTransformation}
	{{{\left| {{\bm{h}}_k^H\left( {\rho_1 {{\bm{h}}_k} + \rho_2 \chi {\bm{h}}_{k,k'}^{\bot} } \right)} \right|}^2}} = {\rho_1^2}{\left\| {{\bm{h}}_k^H} \right\|^4} + {\rho_2^2}{\left| {{\bm{h}}_k^H{\bm{h}}_{k,k'}^{\bot} } \right|^2} + 2\rho_1\rho_2{\left\| {{\bm{h}}_k^H} \right\|^2}{\mathop{\rm Re}\nolimits} \{ {\bm{h}}_k^H \chi {\bm{h}}_{k,k'}^{\bot} \}.
\end{equation}
It can be easily proved that the constraint in (P3) is met with equality for the optimal solution
since otherwise ${\bm{w}}_k$ can be always increased to improve the objective value until the constraint in becomes active. Hence, the constraint in (P3) can be transformed into 
\begin{equation}\label{PowerTransformation}
	{\rho_1^2}{\left\| {{\bm{h}}_k^H} \right\|^2}{{ + }}{\rho_2^2}{\left\| {{\bm{h}}_{k,k'}^{\bot} } \right\|^2} + 2\rho_1\rho_2{\mathop{\rm Re}\nolimits} \{ {\bm{h}}_k^H \chi {\bm{h}}_{k,k'}^{\bot} \}  = {p_k}.
\end{equation}
By plugging (\ref{PowerTransformation}) into (\ref{ObjectiveTransformation}), we can show
\begin{equation}
	\begin{aligned}
		&{\rho_1^2}{\left\| {{\bm{h}}_k^H} \right\|^4} + {\rho_2^2}{\left| {{\bm{h}}_k^H{\bm{h}}_{k,k'}^{\bot} } \right|^2} + 2\rho_1\rho_2{\left\| {{\bm{h}}_k^H} \right\|^2}{\mathop{\rm Re}\nolimits} \{ {\bm{h}}_k^H \chi {\bm{h}}_{k,k'}^{\bot} \} \\
		=& {\rho_1^2}{\left\| {{\bm{h}}_k^H} \right\|^4} + \frac{{{p_k} - {\rho_1^2}{{\left\| {{\bm{h}}_k^H} \right\|}^2}}}{{{{\left\| {{\bm{h}}_{k,k'}^{\bot} } \right\|}^2}}}{\left| {{\bm{h}}_k^H{\bm{h}}_{k,k'}^{\bot} } \right|^2} + 2\rho_1\rho_2  {\left( {{{\left\| {{\bm{h}}_k^H} \right\|}^2} - \frac{{{{\left| {{\bm{h}}_k^H{\bm{h}}_{k,k'}^{\bot} } \right|}^2}}}{{{{\left\| {{\bm{h}}_{k,k'}^{\bot} } \right\|}^2}}}} \right)}{\left\| {{\bm{h}}_k^H} \right\|^2}{\mathop{\rm Re}\nolimits} \{ {\bm{h}}_k^H\chi {\bm{h}}_{k,k'}^{\bot} \},
	\end{aligned}
\end{equation}
where ${{{\left\| {{\bm{h}}_k^H} \right\|}^2} - \frac{{{{\left| {{\bm{h}}_k^H{\bm{h}}_{k,k'}^{\bot} } \right|}^2}}}{{{{\left\| {{\bm{h}}_{k,k'}^{\bot} } \right\|}^2}}}} \ge 0 $. Hence, the objective function in (\ref{P1-TDProof}) achieves its maximum value when ${\bm{h}}_k^H\chi {\bm{h}}_{k,k'}^{\bot} = |\chi| |{\bm{h}}_k^H {\bm{h}}_{k,k'}^{\bot}|$, i.e.,  $\chi = \frac{\left({\bm{h}}^H_k{\bm{h}}_{k,k'}^{\bot}\right)^H}{|{\bm{h}}^H_k{\bm{h}}_{k,k'}^{\bot}|}$. Then, $\rho_2$ can be obtained by solving
\begin{equation}
	{\rho_1^2}{\left\| {{\bm{h}}_k^H} \right\|^2}{{ + }}{\rho_2^2}{\left\| {{\bm{h}}_{k,k'}^{\bot} } \right\|^2} + 2\rho_1\rho_2 | {\bm{h}}_k^H {\bm{h}}_{k,k'}^{\bot} |  = {p_k}.
\end{equation}
i.e., ${\bm{w}}_k^* = \rho_1  {\bm{h}}_k  + \rho_2 \frac{{\bm{h}}_{k,k'}^{\bot}}{\|{\bm{h}}_{k,k'}^{\bot}\|} {\cos \psi _{h_{k'}^{\bot},h_k}}$,
where $\arccos \angle \psi_{H_{k'}^{\bot},h_k} = \frac{\left({\bm{h}}^H_k{\bm{h}}_{k,k'}^{\bot}\right)^H}{|{\bm{h}}^H_k{\bm{h}}_{k,k'}^{\bot}|}$, $\rho_1 = \frac{\sqrt{\varepsilon}}{|{\bm{h}}^H_{k'} {\bm{h}}_k|}$, and $\rho_2 = \frac{{ - \rho_1\left| {{\bm{h}}_k^H{\bm{h}}_{k,k'}^{\bot} } \right| + \sqrt {{{\rho_1^2}}{{\left| {{\bm{h}}_k^H{\bm{h}}_{k,k'}^{\bot} } \right|}^2} - {{\left\| {{\bm{h}}_{k,k'}^{\bot} } \right\|}^2}\left( {{\rho_1^2}{{\left\| {{\bm{h}}_k^H} \right\|}^2} - p_k} \right)} }}{{{{\left\| {{\bm{h}}_{k,k'}^{\bot} } \right\|}^2}}}$.
By combing the above results, the proof is completed.

\footnotesize  	
\bibliography{mybibfile}
\bibliographystyle{IEEEtran}

\end{spacing}
\end{document}